\documentclass[12pt, a4paper]{article}
\usepackage{lipsum}
\usepackage{amsmath}
\usepackage{amssymb}
\usepackage{amsthm}
\usepackage{mathtools}
\usepackage{fancyhdr}
\usepackage{graphicx}
\usepackage{verbatim}
\usepackage{latexsym}
\usepackage{mathrsfs}
\usepackage{mathdots}

\usepackage{tikz}
\usepackage{tikz-cd}

\usepackage{luatex85} 
\usepackage[all]{xy}

\usepackage{float}
\usepackage{authblk}
\usepackage{eqparbox}

\usepackage{array}   
\usepackage{mathtools}
\usepackage{cellspace}
\usepackage{colortbl}

\colorlet{shade}{black!10}

\usepackage{todonotes}
\setuptodonotes{inline}

\usepackage[shortlabels]{enumitem}
\usepackage{pdfpages}
\usepackage{xcolor}
\usepackage[calcwidth]{titlesec}
\usepackage{stmaryrd}
\usepackage{wallpaper}
\usepackage{setspace}
\usepackage[top = 2 cm, bottom = 3 cm, left = 2 cm, right = 2 cm]{geometry}
\usepackage{bbm}

\usepackage[UKenglish]{isodate}
\cleanlookdateon

\usepackage{anyfontsize}
\usepackage{pgf,pgfplots}
\usepackage{textcomp}
\usepackage[margin = 1cm]{caption}
\usepackage{subcaption}
\usepackage{stackengine}

\usepackage{wasysym}
\usepackage{esint} 

\usepackage[style=alphabetic,
  sorting = nyt,
  sortcites=true,
  giveninits=true,
  date=year,
  isbn=false,
  maxbibnames=99,
  maxalphanames=6,
  backend=biber]{biblatex}
\DeclareFieldFormat*{titlecase}{\MakeSentenceCase{#1}}

\DeclareSourcemap{\maps[datatype=bibtex]{\map{\step[fieldset=shorthand, null]}}}

\DeclareSourcemap{
  \maps[datatype=bibtex]{
    \map[overwrite]{
      \step[fieldsource=doi, final]
      \step[fieldset=url, null]
      \step[fieldset=eprint, null]
    }  
  }
}

\DeclareSourcemap{
  \maps[datatype=bibtex, overwrite]{
    \map{
      \step[fieldset=language, null]
      \step[fieldset=month, null]
      \step[fieldset=urldate, null]
    }
  }
}

\AtEveryBibitem{%
  \clearlist{language}%
  \clearlist{month}%
  \clearlist{urldate}%
  \clearfield{urldate}%
  \ifentrytype{inproceedings}
    {\clearlist{booktitle}%
     \clearlist{publisher}%
     \clearlist{location}}
}

\renewbibmacro{in:}{}

\DeclareFieldFormat[misc]{title}{\mkbibquote{#1}}
\DeclareFieldFormat[article,inproceedings]{volume}{\mkbibbold{#1}}
\DeclareFieldFormat[inproceedings]{series}{\mkbibitalic{#1}}

\usepackage[hidelinks,unicode,breaklinks=true]{hyperref}
\usepackage[capitalise,noabbrev,sort]{cleveref}

\usepackage{bookmark}

\usepackage{seqsplit}
\usepackage[notcite,notref,final]{showkeys}

\pgfplotsset{compat=1.16}
\usetikzlibrary{arrows}
\usetikzlibrary{arrows.meta}
\usetikzlibrary{patterns}
\usetikzlibrary{calc}
\usetikzlibrary{math} 

\usetikzlibrary{positioning,decorations.pathreplacing} 

\usepackage{chngcntr}
\counterwithin{figure}{section}
\counterwithin{equation}{section}
\counterwithin{table}{section}

\newcommand{\R}{\mathbb{R}}

\newcommand{\N}{\mathbb{N}}
\newcommand{\Z}{\mathbb{Z}}

\renewcommand{\S}{\mathbb{S}}

\newcommand{\mcN}{\mathcal{N}}

\newcommand{\mcL}{\mathcal{L}}

\newcommand{\mcS}{\mathcal{S}}
\newcommand{\mcG}{\mathcal{G}}
\newcommand{\mcC}{\mathcal{C}}
\newcommand{\mcT}{\mathcal{T}}
\newcommand{\mcH}{\mathcal{H}}
\newcommand{\mcF}{\mathcal{F}}
\newcommand{\mcV}{\mathcal{V}}
\newcommand{\mcD}{\mathcal{D}}
\newcommand{\mcA}{\mathcal{A}}

\newcommand{\x}{\vec{x}}
\newcommand{\y}{\vec{y}}

\newtheoremstyle{theorems}
  {3pt}
  {3pt}
  {\itshape}
  {}
  {\bfseries}
  {.}
  { }
  {}

\newtheoremstyle{proofparts}
  {3pt}
  {0pt}
  {}
  {\parindent}
  {\scshape}
  {:}
  {\newline}
  {}

\newtheoremstyle{claims}
  {2pt}
  {2pt}
  {}
  {\parindent}
  {\bfseries}
  {.}
  { }
  {}

\theoremstyle{theorems}
\newtheorem{thm}{Theorem}[section]

\newtheorem{lemma}[thm]{Lemma}
\newtheorem*{lemma*}{Lemma}
\newtheorem{cor}[thm]{Corollary}
\newtheorem{prop}[thm]{Proposition}

\theoremstyle{definition}
\newtheorem{defn}[thm]{Definition}

\newtheorem{remark}[thm]{Remark}
\newtheorem{notation}[thm]{Notation}

\theoremstyle{proofparts}

\makeatletter
\@addtoreset{proofpart}{thm}
\makeatother

\theoremstyle{claims}

\newtheorem*{claim*}{Claim}

\crefname{thm}{theorem}{theorems}
\crefname{problem}{problem}{problems}
\crefname{lemma}{lemma}{lemmas}
\crefname{cor}{corollary}{corollaries}
\crefname{prop}{proposition}{propositions}
\crefname{conj}{conjecture}{conjectures}
\crefname{defn}{definition}{definitions}
\crefname{note}{note}{notes}
\crefname{ex}{example}{examples}
\crefname{remark}{remark}{remarks}
\crefname{notation}{notation}{notations}
\crefname{assumption}{assumption}{assumptions}
\crefname{claim}{claim}{claims}
\crefname{claim*}{claim}{claims}

\makeatletter
\newcommand{\Biggg}{\bBigg@{3}}
\newcommand{\vast}{\bBigg@{4}}
\newcommand{\Vast}{\bBigg@{5}}
\makeatother

\newcommand{\abs}[1]{\left\vert #1 \right\vert}

\DeclareMathOperator{\Tr}{Tr}

\DeclareMathOperator{\Li}{Li}

\definecolor{emphcolor}{rgb}{0,0,1}           

\newcommand{\expect}[1]{\left\langle #1 \right\rangle}

\newcommand{\ud}{\,\textnormal{d}}
\newcommand{\dd}[1]{\frac{\textnormal{d}}{\textnormal{d} #1}}

\newcommand{\bra}[1]{\left\langle #1 \right\vert}
\newcommand{\ket}[1]{\left\vert #1 \right\rangle}

\let\epsilon\varepsilon
\let\varepsilon\epsilon

\let\eps\epsilon

\title{Pressure of a dilute spin-polarized Fermi gas: Lower bound}
\author{Asbjørn Bækgaard Lauritsen\thanks{\href{mailto:alaurits@ist.ac.at}{\nolinkurl{alaurits@ist.ac.at}}}\,\,}
\author{Robert Seiringer\thanks{\href{mailto:robert.seiringer@ist.ac.at}{\nolinkurl{robert.seiringer@ist.ac.at}}}}
\affil{Institute of Science and Technology Austria, Am Campus 1, 3400 Klosterneuburg, Austria}

\begin{document}
\maketitle

\begin{abstract}
We consider a dilute spin-polarized Fermi gas at positive temperature 
in dimensions $d\in\{1,2,3\}$.
We show that the pressure of the interacting gas is bounded from below by that of the free gas 
plus, to leading order, an explicit term of order $a^d\rho^{2+2/d}$, where $a$ is the $p$-wave scattering length of the repulsive interaction 
and $\rho$ is the particle density.
The results are valid for a wide range of repulsive interactions, including that of a hard core,
and uniform in temperatures at most of the order of the Fermi temperature.
A central ingredient in the proof is a rigorous implementation 
of the fermionic cluster expansion of 
Gaudin, Gillespie and Ripka (Nucl. Phys. A, 176.2 (1971), pp. 237–260).
\end{abstract}

\section{Introduction}
The study of dilute quantum gases has received much interest from the mathematical physics community in the recent decades. 
In particular much work has been done pertaining to the ground state energies of both Fermi and Bose gases in the thermodynamic limit.

For Bose gases in $3$ dimensions the leading term of the ground state energy was first shown by Dyson \cite{Dyson.1957} as an upper bound 
and Lieb--Yngvason \cite{Lieb.Yngvason.1998} as a lower bound. 
The leading term depends only on density and the $s$-wave scattering length of the interaction. 
More recently the second order correction, known as the Lee--Huang--Yang correction, was shown \cite{Yau.Yin.2009,Fournais.Solovej.2020,Fournais.Solovej.2022}.
Also the $2$-dimensional \cite{Lieb.Yngvason.2001,Fournais.Girardot.ea.2022} and $1$-dimensional \cite{Agerskov.Reuvers.ea.2022,Agerskov.2023} settings  
have been studied.

The fermionic setting has been similarly studied in 
the $3$-dimensional 
\cite{Giacomelli.2022a,Falconi.Giacomelli.ea.2021,Lieb.Seiringer.ea.2005,Lauritsen.2023,Lauritsen.Seiringer.2023},
$2$-dimensional \cite{Lieb.Seiringer.ea.2005,Lauritsen.Seiringer.2023} and $1$-dimensional \cite{Lauritsen.Seiringer.2023,Agerskov.Reuvers.ea.2022,Agerskov.2023} 
case.
For fermions the spin is important. For non-zero spin, the leading correction to the energy of the free gas 
is similar to the leading term for bosons and depends only on the density and the $s$-wave scattering length of the interaction.
For spin-polarized (i.e., effectively spin-$0$) fermions the behavior is different. 
The leading correction to the energy of the free gas depends on the $p$-wave scattering length of the interaction instead
and is much smaller for dilute gases, which makes its analysis significantly harder.

A natural question to consider is the extension of these results on the ground state energy to 
positive temperature. This has been done both for bosons \cite{Yin.2010,Seiringer.2008,Haberberger.Hainzl.ea.2023,Mayer.Seiringer.2020,Deuchert.Mayer.ea.2020}
and non-zero spin fermions \cite{Seiringer.2006}.
In this paper we consider the extension for spin-polarized fermions.
More precisely, we consider the problem of finding the pressure $\psi(\beta,\mu)$ at positive temperature $T = 1/\beta$ and chemical potential $\mu$
in the setting of a spin-polarized Fermi gas. 
We are interested in the dilute limit $a^d\rho \ll 1$, 
where $a$ denotes the $p$-\emph{wave scattering length} of the interaction and $\rho$ denotes the particle density.
In this dilute limit we show the lower bound in dimensions $d\in \{1,2,3\}$
\begin{equation*}
\psi(\beta, \mu) \geq \psi_0(\beta,\mu) - c_d(\beta\mu) a^d \rho^{2+2/d}(1 + o(1))
\qquad \textnormal{as } a^d \rho \to 0,
\end{equation*}
for an explicit (temperature dependent) coefficient $c_d(\beta\mu)$.
Here $\psi$ respectively $\psi_0$ denote the pressure of the interacting respectively non-interacting system
at inverse temperature $\beta$ and chemical potential $\mu$.

The term $c_d(\beta\mu) a^d \rho^{2+2/d}$ arises naturally from the two-body interaction and the fact that 
the two-body density vanishes quadratically for incident particles. 
In the low-temperature limit $\beta\mu \to \infty$ the coefficients $c_d(\beta\mu)$ converge to 
the corresponding zero-temperature constants \cite{Lauritsen.Seiringer.2023,Agerskov.Reuvers.ea.2022}.
The temperature dependence of this term can then be understood via the temperature dependence of the two-particle density 
of the free state.

The result is valid for temperatures $T$ at most of the order of the Fermi temperature $T_F \sim \rho^{2/d}$ of the free gas.
For larger temperatures one should expect that thermal effects become larger than quantum effects,
and thus the gas should behave more like a (high temperature) classical gas. 
The natural parameter capturing the temperature is the \emph{fugacity} $z=e^{\beta\mu}$. 
In terms of the fugacity the constraint that the temperature satisfies $T\lesssim T_F$ reads $z \gtrsim 1$.

Our method of proof consists of  computing the pressure of a Jastrow-type trial state 
using a rigorous implementation \cite{Lauritsen.Seiringer.2023,Lauritsen.2023} of the fermionic cluster expansion of Gaudin--Gillepie--Ripka 
\cite{Gaudin.Gillespie.ea.1971}.
A similar method was employed in the zero-temperature setting \cite{Lauritsen.Seiringer.2023}, 
with the important difference that, because of the smoothness of the momentum distribution, 
the condition for convergence we obtain at positive temperature is uniform in the volume.
Thus we can compute the thermodynamic limit directly, without appealing to a box method 
of localizing a trial state into large but finite boxes as done in \cite{Lauritsen.Seiringer.2023}.

\subsection{Precise statement}
To state our main theorem precisely, define the (spin-polarized) fermionic Fock space 
$\mcF = \bigoplus_{n=0}^\infty L_a^2\left([0,L]^{3n}\right) = \bigoplus_{n=0}^\infty \bigwedge^{n} L^2\left([0,L]^{3}\right)$.
On this space we define the free Hamiltonian $\mcH$, the number operator $\mcN$ and interaction operator $\mcV$ as follows
(in natural units where $\frac{\hbar}{2m}=1$)
\begin{equation*}
\begin{aligned}
\mcH & = (0, H_{1}, \ldots, H_{n}, \ldots),
\qquad
  & 
  H_n & = \sum_{j=1}^n -\Delta_{x_j},
\\
\mcN & = (0, 1, \ldots, n, \ldots),
\\
  \mcV & = (0,0,V_2,\ldots,V_n,\ldots),
  & 
  V_n & = \sum_{1\leq i < j \leq n} v(x_i-x_j).
\end{aligned}
\end{equation*}
The interacting Hamiltonian is then $\mcH + \mcV$.
In the calculations below we will use periodic boundary conditions for convenience. 
The thermodynamical quantities don't depend on the choice of boundary conditions \cite{Robinson.1971}
and hence we are free to choose the most convenient ones.
We are interested in determining the pressure of the system described by this Hamiltonian at inverse temperature $\beta$ and chemical potential $\mu$.
We denote this by 
\begin{equation*}
	\psi(\beta, \mu) = \lim_{L\to \infty}
  \sup_{\varGamma} 
    P[\varGamma],
    \qquad 
     -L^d P[\varGamma] = 
     \Tr_{\mcF}\left[(\mcH - \mu \mcN + \mcV)\varGamma\right]
    - \frac{1}{\beta} S(\varGamma),
\end{equation*}
where $S(\varGamma) = - \Tr \varGamma \log \varGamma$ is the entropy of the state $\varGamma$ and 
$P[\varGamma]$ is the pressure functional.
By \emph{state} we mean a density matrix, i.e., a positive trace-class operator on $\mcF$ of unit trace.
(We suppress from the notation the dependence on the dimension $d$ and the length $L$.)
We denote moreover by
\begin{equation*}
  \psi_0(\beta, \mu) = \lim_{L\to \infty}
  \sup_{\varGamma} 
    P_0[\varGamma],
    \qquad 
     -L^d P_0[\varGamma] = 
    \Tr_{\mcF}\left[(\mcH - \mu \mcN)\varGamma\right]
    - \frac{1}{\beta} S(\varGamma),
\end{equation*}
the pressure and pressure functional of the free gas.
The supremum is a maximum and is achieved for the Gibbs state 
\begin{equation}\label{eqn.define.Gamma}
	\varGamma = Z^{-1} \exp \left({-\beta (\mcH - \mu \mcN)}\right) = Z^{-1} (\varGamma_0, \varGamma_1,\ldots,\varGamma_n, \ldots),
	\qquad 
	\varGamma_n = e^{\beta \mu n} e^{-\beta H_n}.
\end{equation}
Then \cite[Equation (8.63)]{Huang.1987}
\begin{equation}\label{eqn.formula.psi0}
\begin{aligned}
\psi_0(\beta,\mu) 
	& = \lim_{L\to \infty} \frac{1}{L^d} 
    \left[-\Tr_{\mcF}\left[(\mcH - \mu \mcN)\varGamma\right] 
    + \frac{1}{\beta} S(\varGamma)\right]
	= \lim_{L\to\infty} \frac{1}{L^d\beta} \log Z
  \\ & 
  = \frac{1}{\beta (2\pi)^d}  \int_{\R^d} \log\left(1 + e^{\beta\mu - \beta|k|^2}\right) \ud k.
\end{aligned}
\end{equation}
To state our main theorem we moreover define the $p$-\emph{wave scattering length} $a$.
(See also \cite[Appendix A]{Lieb.Yngvason.2001} and \cite[Equations (2.9), (4.3)]{Seiringer.Yngvason.2020}.)
\begin{defn}[{\cite[Definitions 1.1, 1.6 and 1.8]{Lauritsen.Seiringer.2023}}]
\label{def.scattering.length}
The $p$-\emph{wave scattering length} $a$ of the interaction $v$ in dimension $d$ is defined by 
\begin{equation*}
	c_d a^d = \inf \left\{\int_{\R^d} \left(\abs{\nabla f_0(x)}^2 + \frac{1}{2} v(x) f_0(x)^2\right)|x|^2 \ud x
		: f_0(x) \to 1 \textnormal{ for } |x|\to \infty \right\},
\end{equation*}
where 
\begin{equation}\label{eqn.define.cd}
  c_d = \begin{cases}
  12\pi & d=3, \\ 4\pi & d=2, \\ 2 & d=1.
  \end{cases}
\end{equation}
The minimizer $f_0$ is the $p$-\emph{wave scattering function}. 
(If $v(x) = +\infty$ for some $x$ we interpret $v(x)\ud x$ as a measure.
We suppress from the notation the dependence of $a$ and $f_0$ on the dimension $d$.)
\end{defn}

\noindent
The dimensionless parameter measuring the diluteness is then $a^d \rho$, with $\rho$ the particle density
given by $\rho = \partial_\mu \psi(\beta,\mu)$.
We are interested in a dilute limit, meaning that $a^d\rho \ll 1$. 
Moreover, we are considering temperatures $T\lesssim T_F \sim \rho^{2/d}$ meaning that $z \gtrsim 1$.
As mentioned in the introduction, small $z$ corresponds to a (high-temperature) classical gas.


We shall prove the following theorem.
\begin{thm}\label{thm.main}
Let $v\geq 0$ be radial and of compact support. 
If $d=1$ assume moreover that $\int \left(\abs{\partial f_0}^2 + \frac{1}{2} v f_0^2\right) \ud x < \infty$. 
For any $z_0 > 0$ there exists $c > 0$ such that if $a^d\rho_0 < c$ then, uniformly in $z = e^{\beta \mu} \geq z_0$,
we have the lower bound
\begin{equation*}
  \psi(\beta,\mu) 
    \geq \psi_0(\beta,\mu) 
    - 2\pi c_d \frac{-\Li_{d/2+1}(-z)}{(-\Li_{d/2}(-z))^{1+d/2}} a^d \rho_0^{2+2/d}
    \left[
    1 + \delta_d
    \right],
\end{equation*}
where $\rho_0=\partial_\mu \psi_0(\beta,\mu)$ is the particle density of the free gas, 
the constants $c_d$ are defined in \Cref{eqn.define.cd} and
\begin{equation}\label{eqn.errors.thm.main}
\abs{\delta_d} \leq \begin{cases}
C(a^3\rho_0)^{1/39} \abs{\log a^3\rho_0}^{12/13} & d=3,
\\
C(a^2\rho_0)^{1/5} \abs{\log a^2\rho_0}^{8/7} & d=2, 
\\
C(a\rho_0)^{1/7} \abs{\log a\rho_0}^{12/7} & d=1.
\end{cases}
\end{equation}
\end{thm}

\noindent
Here $\Li_s$ denotes the \emph{polylogarithm}. It satisfies 
\cite[Equation 25.12.16]{dlmf}
\begin{equation}\label{eqn.polylog}
-\Li_{s}(-e^x) = \frac{1}{\Gamma(s)} \int_0^\infty \frac{t^{s-1}}{e^{t-x}+1} \ud t
\end{equation}
with $\Gamma$ the Gamma-function.

We expect that the lower bound of \Cref{thm.main} is in fact an equality (with a potentially different bound on the error-term).
It remains an open problem to prove this.

\begin{remark}
For better comparison with the zero-temperature result in \cite{Lauritsen.Seiringer.2023}, we find it convenient 
to write the correction to the pressure of the free gas in terms of the particle density (of the free gas) $\rho_0$.
The latter is given explicitly as 
\begin{equation}
\rho_0 =  - \frac{1}{(4\pi \beta)^{d/2}}  \Li_{d/2}(-z)
    \label{eqn.rho1.free.TL}
\end{equation}
This follows from an elementary computation, which we give in \Cref{lem.prop.rho0} below.
\end{remark}

To leading order $\rho \simeq \rho_0$. More precisely

\begin{cor}\label{prop.rho=rho0}
Under the same assumptions as in \Cref{thm.main} we have for the particle density $\rho = \partial_\mu \psi(\beta,\mu)$
\begin{equation*}
\rho = \rho_0\left[1 + O( (a^d\rho_0)^{1/2})\right].
\end{equation*}
\end{cor}

We shall give the proof at the end of this section.
In particular the conditions of small $a^d\rho$ and of small $a^d\rho_0$ are equivalent. 
Moreover, the error-terms of \Cref{thm.main} can equally well be written with $\rho_0$ replaced by $\rho$.

\begin{remark}
The additional assumption on $v$ in dimension $d=1$ is discussed in \cite[Remark 1.10]{Lauritsen.Seiringer.2023}. 
If $v$ is either smooth or has a hard core (meaning that $v(x)=+\infty$ for $|x|\leq a_0$ for some $a_0 > 0$) 
this assumption is satisfied.
\end{remark}

\begin{remark}
The term of order $a^d \rho_0^{2+2/d}$ depends on the temperature. 
This is different to the setting of spin-$\frac{1}{2}$ fermions,
where the analogous term (in $3$ dimensions) is $2\pi a \rho_0^2$ \cite{Seiringer.2006} uniformly in the temperature.
That the term of order $a^d \rho_0^{2+2/d}$ should depend on the temperature may be heuristically understood as follows:
This term arises from the fact that the two-body density vanishes quadratically for incident particles. 
The rate at which it vanishes depends on the exact state, and thus the temperature. 
Concretely, the two-particle density of the free gas satisfies 
\begin{equation}
\begin{aligned}
\rho^{(2)}(x_1,x_2)
  & = 2\pi \frac{- \Li_{d/2+1}(-z)}{(-\Li_{d/2}(-z))^{1+2/d}} \rho_0^{2+2/d} |x_1-x_2|^2
    \left[1 
    +  O\left(\rho_0^{2/d}|x_1-x_2|^2\right)\right],
\end{aligned}
\label{eqn.rho2.free.TL}
\end{equation}
where $O\left(\rho_0^{2/d}|x_1-x_2|^2\right)$ is understood as being bounded by $C \rho_0^{2/d}|x_1-x_2|^2$ uniformly.
This follows from an elementary computation, which we give in \Cref{lem.prop.rho0} below.

In the low-temperature limit $z \to \infty$
we recover the zero-temperature constants in the terms of order $a^d \rho_0^{2+2/d}$.
Namely, we claim that
\begin{equation}\label{eqn.constant.low.temperature}
2\pi c_d \frac{- \Li_{d/2+1}(-z)}{(-\Li_{d/2}(-z))^{1+2/d}}
= 
\begin{cases}
\begin{aligned}
& \textstyle \frac{12\pi}{5}(6\pi^2)^{2/3} \hspace*{-1em} && + \textstyle O\left((\log z)^{-2}\right) & d=3,
\\
& \textstyle 4\pi^2 && + \textstyle O\left((\log z)^{-2}\right) & d=2,
\\
& \textstyle \frac{2\pi^2}{3} && + \textstyle O\left((\log z)^{-2}\right) & d=1,
\end{aligned}
\end{cases} 
\quad 
\textnormal{as } 
z\to \infty.
\end{equation}
To see this write (following \cite{Wood.1992})
\begin{multline*}
-\Li_{s}(-e^x) 
  = \frac{1}{\Gamma(s)} \int_0^\infty \frac{t^{s-1}}{e^{t-x}+1} \ud t
  = \frac{1}{\Gamma(s)} \left[\int_0^x t^{s-1} \ud t - \int_0^x \frac{t^{s-1}}{e^{x-t}+1} \ud t + \int_x^\infty \frac{t^{s-1}}{e^{t-x}+1} \ud t\right]
	\\
	= \frac{x^{s}}{\Gamma(s+1)} 
		- \frac{1}{\Gamma(s)} \int_0^x \frac{(x-u)^{s-1} - (x+u)^{s-1}}{e^u + 1} \ud u 
		- \frac{1}{\Gamma(s)} \int_x^\infty \frac{(x+u)^{s-1}}{e^u+1} \ud u
\end{multline*}
where we changed variables $t= x\pm u$. 
The middle  and last integrals can easily be bounded as $O(x^{s-2})$ and $O(x^s e^{-x})$ respectively. 
Thus 
\begin{equation}\label{eqn.asym.polylog}
\begin{aligned}
-\Li_{s}(-e^x) 
	& 
  = \frac{x^{s}}{\Gamma(s+1)} 
		+ O(x^{s-2}),
\end{aligned}
\end{equation}
and \Cref{eqn.constant.low.temperature} follows.
\end{remark}

\begin{remark}
The error bounds in \Cref{thm.main} are uniform in $z$. 
They arise as the worst cases of two types of bounds, one good for $z\sim 1$ and one good for $z \gg 1$. 
In particular, for concrete values of $z$, the error bounds can be improved. 
See \Cref{lem.pressure.low.temp,lem.pressure.high.temp} below.
\end{remark}

Finally we give the 
\begin{proof}[{Proof of \Cref{prop.rho=rho0}}]
Note that $\psi(\beta,\mu)$ is a convex function of $\mu$. Thus we may bound its derivative by any difference quotient.
More precisely for any $\eps > 0$ we have 
\begin{equation*}
\rho = \partial_\mu \psi(\beta,\mu) \leq \frac{\psi(\beta,\mu+\eps) - \psi(\beta,\mu)}{\eps}.
\end{equation*}
Using the trivial upper bound $\psi(\beta,\mu+\eps)\leq \psi_0(\beta,\mu+\eps)$ (which is a consequence of the assumed non-negativity of the interaction potential $v$) and the lower bound of \Cref{thm.main}
we conclude that 
\begin{equation*}
\rho \leq \frac{\psi_0(\beta,\mu+\eps) - \psi_0(\beta,\mu)}{\eps} + C a^d\rho_0^{2+2/d} \eps^{-1}
  = \rho_0 + O\left(\abs{\partial_\mu^2 \psi_0} \eps\right) + O\left(a^d\rho_0^{2+2/d} \eps^{-1}\right).
\end{equation*}
Using the explicit formula for $\rho_0=\partial_\mu \psi_0$ and optimising in $\eps$ we get that $\rho \leq \rho_0(1 + O( (a^d\rho_0)^{1/2}))$.
For $\eps < 0$ the argument is analogous only the direction of the inequalities is reversed.
\end{proof}

\subsection{Strategy of proof}
To prove \Cref{thm.main} we distinguish two cases. That of a ``low-temperature'' setting and that of a ``high-temperature'' setting. 
For sufficiently small temperatures we compare to the ground state energy studied in \cite{Lauritsen.Seiringer.2023}.
For larger temperatures we consider a specific trial state $\varGamma_J$ of Jastrow-type (defined in \Cref{eqn.def.GammaJ} below) 
and compute the pressure functional evaluated on this trial state.
For these computations we use a rigorous implementation \cite{Lauritsen.Seiringer.2023,Lauritsen.2023} of the formal
cluster expansion of Gaudin--Gillepie--Ripka \cite{Gaudin.Gillespie.ea.1971}.

Temperature-dependent errors naturally arise as powers of 
$\zeta := 1 + \abs{\log z}.$
We shall prove the following propositions.
\begin{prop}\label{lem.pressure.low.temp}
Let $v\geq 0$ be radial and of compact support. 
If $d=1$ assume moreover that $\int \left(\abs{\partial f_0}^2 + \frac{1}{2} v f_0^2\right) \ud x < \infty$. 
Then for sufficiently small $a^d\rho_0$ and large $z = e^{\beta\mu}$ we have 
\begin{equation}\label{eqn.pressure.low.temp}
\psi(\beta,\mu) 
  \geq \psi_0(\beta,\mu) - 2\pi c_d \frac{-\Li_{d/2+1}(-z)}{(-\Li_{d/2}(-z))^{1+2/d}} a^d \rho_0^{2+2/d} 
  \left[1 + \delta_d\right]
\end{equation}
where $\rho_0$ is the particle density of the free gas, $c_d$ is defined in \Cref{eqn.define.cd} and 
\begin{equation}\label{eqn.errors.zero.temp}
\abs{\delta_d}
  \lesssim
  \begin{cases}
  \begin{aligned}
  &a^2\rho_0^{2/3} && + (a^3\rho_0)^{-1} \zeta^{-2} & d=3,
  \\
  &a^2\rho_0 \abs{\log a^2\rho_0}^2 \hspace*{-1em}&& + (a^2\rho_0)^{-1} \zeta^{-2}  & d=2,
  \\
  &(a\rho_0)^{13/17} && + (a\rho_0)^{-1} \zeta^{-2}  & d=1.
  \end{aligned}
  \end{cases}
\end{equation}
\end{prop}

\begin{prop}\label{lem.pressure.high.temp}
Let $v\geq 0$ be radial and of compact support. 
If $d=1$ assume moreover that $\int \left(\abs{\partial f_0}^2 + \frac{1}{2} v f_0^2\right) \ud x < \infty$. 
Then for $z = e^{\beta\mu}$ satisfying $z \gtrsim 1$ there exists a constant $c>0$ such that if $a^d\rho_0 < c$ and 
$a^d \rho_0 \zeta^{d/2} \abs{\log a^d\rho_0} < c$ then 
\begin{equation*}
  \psi(\beta,\mu) 
    \geq \psi_0(\beta,\mu) 
    - 2\pi c_d \frac{-\Li_{d/2+1}(-z)}{(-\Li_{d/2}(-z))^{1+d/2}} a^d \rho_0^{2+2/d}
    \left[
    1 + \delta_d
    \right],
\end{equation*}
where $\rho_0$ is the particle density of the free gas, $c_d$ is defined in \Cref{eqn.define.cd}
and 
\begin{equation}\label{eqn.delta.errors.high.temp}
\begin{aligned}
  \abs{\delta_d}
    & \lesssim 
    \begin{cases}
    (a^3\rho_0)^{6/15} \zeta^{-3/5} 
  + (a^3\rho_0) \zeta^{1/2} \abs{\log a^3\rho_0}^2 
  + (a^3\rho_0)^{7/3} \zeta^{9/2} \abs{\log a^3\rho_0}^3
  & d=3,
  \\
  (a^2\rho_0)^{1/2} \zeta^{-1/2} 
  + (a^2\rho_0) \zeta \abs{\log a^2\rho_0}
  + (a^2\rho_0)^2 \zeta^{3} \abs{\log a^2\rho_0}^{3},
  & d=2,
  \\
  (a\rho_0)^{1/2} \abs{\log a\rho_0}^{1/2} 
  + a\rho_0 \zeta^{3/2} \abs{\log a\rho_0}^3
  & d=1.
    \end{cases}
\end{aligned}
\end{equation}
\end{prop}

\noindent
\Cref{lem.pressure.low.temp} is a simple corollary of \cite[Theorems 1.3, 1.7, 1.9]{Lauritsen.Seiringer.2023},
extending the result to small positive temperatures. 
\Cref{lem.pressure.high.temp} is the main new result of this paper. 
Most of the rest of the paper is concerned with the proof of \Cref{lem.pressure.high.temp}.
\Cref{thm.main} is an immediate consequence:
\begin{proof}[Proof of \Cref{thm.main}]
We use the lower bound in \Cref{lem.pressure.low.temp} for 
\begin{equation*}
 \zeta
 \gtrsim \begin{cases}
 (a^3\rho_0)^{-20/39} \abs{\log a^3\rho_0}^{-6/13}
 & d=3,
 \\
 (a^2\rho_0)^{-3/5} \abs{\log a^2\rho_0}^{-2/5}
 & d=2,
 \\
 (a\rho_0)^{-4/7} \abs{\log a\rho_0}^{-6/7}
 & d=1
 \end{cases}
\end{equation*}
and the lower bound in \Cref{lem.pressure.high.temp} otherwise. 
\Cref{thm.main} follows.
\end{proof}

\begin{remark}
We expect that with the method presented here one could improve the error bounds in \Cref{lem.pressure.high.temp} (and consequently \Cref{thm.main}) 
slightly by computing the values of more small diagrams in the Gaudin--Gillepie--Ripka expansion precisely. 
See also \cite[Remark 1.5]{Lauritsen.Seiringer.2023}.
This is similar to what is done in \cite{Lauritsen.2023,Basti.Cenatiempo.ea.2022a}. 

More precisely we expect that by doing so one could improve the bounds in \Cref{lem.pressure.high.temp} to
\begin{equation}\label{eqn.optimal.bdds}
	\abs{\delta_d}
	\lesssim 
	O
	\left(
	\begin{cases}
    (a^3\rho_0)^{6/15} \zeta^{-3/5} 
  & d=3
  \\
  (a^2\rho_0)^{1/2} \zeta^{-1/2} 
  & d=2
  \\
  (a\rho_0)^{1/2} \abs{\log a\rho_0}^{1/2} 
  & d=1
    \end{cases}
    \right)
    + O \left((a^d\rho_0)^{-2/d} \left(a^d\rho_0 \zeta^{d/2} \abs{\log a^d\rho_0}\right)^{n}\right)
\end{equation}
for any $n$. 
This would then propagate to better error terms in \Cref{thm.main}.
More precisely, by using the bound in \Cref{lem.pressure.low.temp} for $\zeta \geq \zeta_0$ and the bound in 
\Cref{lem.pressure.high.temp} with error improved as in \Cref{eqn.optimal.bdds} otherwise and optimising in
$\zeta_0$ one would improve the error bound in \Cref{thm.main} to 
\begin{equation*}
	\abs{\delta_d}
		\lesssim 
		\begin{cases}
		C_\eps (a^3\rho_0)^{1/3-\eps} 
		& d=3,
		\\
		(a^2\rho_0)^{1/2} 
		& d=2,
		\\
		(a\rho_0)^{1/2} \abs{\log a\rho_0}^{1/2} 
		& d = 1
		\end{cases}
\end{equation*}
for any $\eps > 0$, where $C_\eps$ depends on $\eps$, by taking $n$ sufficiently large in \Cref{eqn.optimal.bdds}.


The first terms in \Cref{eqn.optimal.bdds} come from the precise evaluation of certain small diagrams.
In dimension $d=2,3$ one should not expect to get better bounds than this using the method presented here.
In dimension $d=1$ one might be able to do a more precise analysis, see \Cref{rmk.compare.zero.temp.d=1}, and thus improve the bound.
\end{remark}

The proof of \Cref{lem.pressure.low.temp} 
will be given in 
in \Cref{sec.low.temp}. 
It is mostly independent of the rest of the paper
 (\Cref{sec.preliminaries,sec.GGR.expansion,sec.calc.terms}) 
which is devoted to the proof of \Cref{lem.pressure.high.temp}.

\paragraph*{Structure of the paper:}
First, in \Cref{sec.low.temp} we give the proof of \Cref{lem.pressure.low.temp}.
Then, in \Cref{sec.preliminaries} we define the trial state $\varGamma_J$ and give some preliminary computations.
Next, in \Cref{sec.GGR.expansion} we compute reduced densities of the trial state $\varGamma_J$ using the 
(rigorous implementation of the) Gaudin--Gillepie--Ripka expansion.
Finally, in \Cref{sec.calc.terms} we calculate the individual terms in the pressure functional and prove \Cref{lem.pressure.high.temp}.
In \Cref{sec.density.Gamma_J} we show that $\varGamma_J$ has particle density $\approx \rho_0$.


\section{Low temperature}\label{sec.low.temp}
In this section we prove \Cref{lem.pressure.low.temp} by comparing to the zero-temperature problem.
\begin{proof}[Proof of \Cref{lem.pressure.low.temp}]
The pressures $\psi, \psi_0$ (of the interacting and non-interacting gas, respectively) 
are the Legendre transforms of the corresponding free energy densities $\phi, \phi_0$. 
That is,
\begin{equation}\label{eqn.legendre.low.temp}
\begin{aligned}
\psi(\beta,\mu) 
  & = \sup_{\tilde\rho} \left[\tilde \rho \mu - \phi(\beta,\tilde\rho)\right]
  \geq \rho_0\mu - \phi(\beta,\rho_0)
\\
  \psi_0(\beta,\mu) 
  & = \sup_{\tilde\rho} \left[\tilde \rho \mu - \phi_0(\beta,\tilde\rho)\right]
  = \rho_0\mu - \phi_0(\beta,\rho_0)
\end{aligned}
\end{equation}
with $\rho_0$ the density of the free gas at chemical potential $\mu$ and inverse temperature $\beta$, given in \Cref{eqn.rho1.free.TL}. 
We may trivially bound the free energy density by the ground state energy density $e$.
The latter is bounded from above in \cite[Theorems 1.3, 1.7 and 1.9]{Lauritsen.Seiringer.2023}. That is,
\begin{equation}\label{eqn.trivial.upper.bdd.low.temp}
\phi(\beta, \rho_0) \leq e(\rho_0) \leq e_0(\rho_0) + c_{0,d} a^d \rho_0^{2+2/d}[1 + \delta_d],
\end{equation}
with $e_0(\rho_0)$ denoting the ground state energy density of the free gas and
\begin{equation*}
\begin{aligned}
c_{0,d}
  & = 
  \begin{cases}
  \frac{12\pi}{5}(6\pi^2)^{2/3} & d=3,
  \\
  4\pi^2 & d=2,
  \\
  \frac{2\pi^2}{3} & d=1,
  \end{cases}
  \qquad
  &
  \abs{\delta_d}
  &
  \lesssim
  \begin{cases}
  a^2\rho^{2/3} & d=3,
  \\
  a^2\rho_0\abs{\log a^2\rho_0}^2 & d=2,
  \\
  (a\rho_0)^{13/17} & d=1.
  \end{cases}
\end{aligned}
\end{equation*}
By a straightforward calculation, the ground state energy density of the free gas is 
\begin{equation*}
\begin{aligned}
e_0(\rho_0) & = 4\pi \frac{d^{2/d}}{d+2} \left(\frac{d}{2}\right)^{2/d} \Gamma(d/2)^{2/d} \rho_0^{1+2/d}.
\end{aligned}
\end{equation*}
By \Cref{eqn.asym.polylog,eqn.rho1.free.TL,eqn.formula.psi0} 
we have for large $z = e^{\beta\mu}$ (see also \cite[Equation (11.31)]{Huang.1987})
\begin{equation*}
\begin{aligned}
  \psi_0(\beta,\mu)
    & = \beta^{-1-d/2} \frac{\abs{\S^{d-1}}\Gamma(d/2)}{2(2\pi)^d} ( - \Li_{d/2+1}(-e^{\beta\mu}))
    \\
    & = 4\pi \rho_0^{1+2/d} \frac{-\Li_{d/2+1}(-e^{\beta\mu})}{(-\Li_{d/2}(-e^{\beta\mu}))^{1+2/d}}
  = \frac{2}{d} e_0(\rho_0) \left(1 + O\left( (\beta\mu)^{-2}\right)\right),
\end{aligned}
\end{equation*}
where $\abs{\S^{d-1}} = \frac{2\pi^{d/2}}{\Gamma(d/2)}$ is the area of the $(d-1)$-sphere.
Thus
\begin{equation*}
\phi_0(\beta,\rho_0) = \rho_0\mu - \psi_0(\beta,\mu)
  = e_0 + O\left(\rho_0^{1+2/d} (\beta\mu)^{-2}\right).
\end{equation*}
Combining this with \Cref{eqn.trivial.upper.bdd.low.temp,eqn.legendre.low.temp} we conclude the proof of \Cref{lem.pressure.low.temp}.
\end{proof}

The rest of the paper concerns the proof of \Cref{lem.pressure.high.temp}.
We start with some preliminary computations.

\section{Preliminaries}\label{sec.preliminaries}
To prove \Cref{lem.pressure.high.temp} we will consider a finite system on a cubic box of side length $L$ with periodic boundary conditions and 
bound $\psi(\beta,\mu)$ from below by the pressure functional evaluated on the trial state 
\begin{equation}\label{eqn.def.GammaJ}
\varGamma_J = \frac{Z}{Z_J} F \varGamma F,
\qquad F = \bigoplus_{n=0}^\infty F_n,
\qquad F_n = \prod_{1 \leq i < j \leq n} f(x_i - x_j),
\end{equation}
where $f$ is some cut-off and rescaled scattering function defined in \Cref{eqn.define.f} below,
where $\varGamma$ is defined in \Cref{eqn.define.Gamma},
and where $Z_J$ is such that this is normalised with $\Tr \varGamma_J = 1$.
Concretely, on the $n$-particle space $\varGamma_J$ acts via the kernel
\begin{equation*}
  Z_J^{-1} F_n(X_n) \varGamma_n(X_n,Y_n) F_n(Y_n).
\end{equation*}
(Recall that $\varGamma$ acts via the kernel $Z^{-1}\varGamma_n(X_n,Y_n)$.)
The function $f$ is more precisely 
\begin{equation}\label{eqn.define.f}
	f(x)  
	= \begin{cases}
	\frac{1}{1 - a^d/b^d} f_0(x) & |x|\leq b 
	\\
	1 & |x|\geq b
	\end{cases}
\end{equation}
where $f_0(x)$ is the $p$-wave scattering function defined in \Cref{def.scattering.length}
and $b$ is a length to be chosen later. We will choose $a \ll b \leq C \rho_0^{-1/d}$.
(Here and in the following $\rho_0$ denotes the particle density of the free gas in finite volume.)
In particular for $a^d\rho_0$ small enough $b$ is larger than the range of $v$ and so $f$ is continuous 
(since $f_0(x) = 1 - \frac{a^3}{|x|^3}$ for $x$ outside the support of $v$).

\begin{notation}
\begin{itemize}
\item[]

\item 
We will denote expectation values of operators in the free state $\varGamma$ by $\expect{\cdot}_0$
and in the trial state $\varGamma_J$ by $\expect{\cdot}_J$.
That is, 
$\expect{\mcA}_0 = \Tr_{\mcF}[\mcA\varGamma]$ and $\expect{\mcA}_J = \Tr_{\mcF}[\mcA\varGamma_J]$ for any operator $\mcA$ on $\mcF$.

\item 
We denote $g(x) = f(x)^2 - 1$.

\item 
For any function $h$ we write $h_e = h_{ij} = h(x_i-x_j)$ for an edge $e = (i,j)$.

\item 
Moreover we write $\gamma^{(1)}_e = \gamma^{(1)}_{ij} = \gamma^{(1)}(x_i; x_j)$ for an edge $e = (i,j)$,
where $\gamma^{(1)}$ is the $1$-particle density matrix of $\varGamma$ defined in \Cref{eqn.define.gamma(q)} below.

\item 
We write $X_n = (x_1,\ldots,x_n)$ and $X_{[n,m]} = (x_n,\ldots,x_m)$ if $n\leq m$. For $n > m$ then $X_{[n,m]} = \varnothing$.
\end{itemize}
\end{notation}

\begin{remark}
The trial state $\varGamma_J$ does not have (average) particle density $\rho_0$. 
However we have that 
\begin{equation}\label{eqn.density.Gamma_J}
  \frac{1}{L^d}\expect{\mcN}_J = \rho_0\left(1 + O(a^d b^2 \rho_0^{1+2/d}) + O\left((a^{d}\rho_0)^2 \zeta^d (\log b/a)^2\right)\right).
\end{equation}
This is not needed for the proof of \Cref{lem.pressure.high.temp}, however. We give the proof of  \eqref{eqn.density.Gamma_J} in \Cref{sec.density.Gamma_J}.
\end{remark}

\noindent
We normalize $q$-particle density matrices of $\varGamma$ as 
\begin{equation}\label{eqn.define.gamma(q)}
\begin{aligned}
\gamma^{(q)}(X_q; Y_q)
	& = \frac{1}{Z} \sum_{n=q}^{\infty} \frac{n!}{(n-q)!} \idotsint \varGamma_n(X_{q}, X_{[q+1,n]}; Y_q, X_{[q+1,n]}) \ud X_{[q+1,n]}.
\end{aligned}
\end{equation}
The state $\varGamma$ is quasi-free and particle preserving. Thus by Wick's rule we have for the $q$-particle density
\begin{equation*}
\rho^{(q)}(X_q) = \gamma^{(q)}(X_q;X_q) = \det \left[\gamma^{(1)}_{ij}\right]_{1\leq i,j \leq q}.
\end{equation*}
Moreover, by translation invariance, we have that $\gamma^{(1)}(x;y)$ is a function of $x-y$ only. 
With a slight abuse of notation we then write 
\begin{equation*}
\gamma^{(1)}(x;y) = \gamma^{(1)}(x-y) = \frac{1}{L^d}\sum_{k\in \frac{2\pi}{L}\Z^d} \hat\gamma^{(1)}(k) e^{-ik(x-y)}.
\end{equation*}
A simple calculation shows that (see \cite[Equation (8.65)]{Huang.1987})
\begin{equation*}
\hat \gamma^{(1)}(k) 
= \frac{ze^{ - \beta|k|^2}}{1 + ze^{- \beta|k|^2}}
= \frac{e^{\beta\mu - \beta|k|^2}}{1 + e^{\beta\mu - \beta|k|^2}}.
\end{equation*}

\noindent
For the proof of \Cref{lem.pressure.high.temp} we compute the pressure of the trial state $\varGamma_J$.
We have 
\begin{equation}\label{eqn.calc.free.energy.initial}
\begin{aligned}
	\psi(\beta,\mu)
	& \geq \limsup_{L\to \infty} \frac{1}{L^d} \left[-\expect{\mcH - \mu \mcN + \mcV}_J + \frac{1}{\beta} S(\varGamma_J)\right]
	\\
	& = \limsup_{L\to \infty} \frac{1}{L^d} 
		\left[-\expect{\mcH}_J - \mu \expect{\mcN}_J - \frac{1}{2}\iint v_{12} \rho^{(2)}_J \ud x_1 \ud x_2 
				+ \frac{1}{\beta} S(\varGamma_J)\right],
\end{aligned}		
\end{equation}
where $\rho^{(2)}_J$ is the two-body reduced density of the trial state $\varGamma_J$.
In general we denote by $\rho^{(q)}_J$ the $q$-particle density of $\varGamma_J$.
We calculate $\rho^{(2)}_J$ in \Cref{sec.GGR.expansion} using the Gaudin--Gillepie--Ripka expansion
and we compute the individual terms of \Cref{eqn.calc.free.energy.initial} in \Cref{sec.calc.terms} below.
First, however, we need some preliminary bounds.

\subsection{Useful bounds}
We recall some useful bounds on the scattering function (defined in \Cref{eqn.define.f}) 
from \cite{Lauritsen.Seiringer.2023}.
\begin{lemma}\label{lem.bdd.int.f}
The scattering function $f$ satisfies
\begingroup
\allowdisplaybreaks
\begin{align}
	\int \abs{1 - f(x)^2} |x|^n \ud x 
		& \leq \begin{cases}
		C a^d \log b/a & n=0
		\\
		C a^d b^n & n > 0		
		\end{cases}
	\label{eqn.Ig.bdd}
	\\
	\int \left(\abs{\nabla f(x)}^2 + \frac{1}{2} v(x) f(x)^2\right)|x|^2 \ud x 
		& = c_d a^d \left(1 + O(a^d/b^d))\right) 
	\label{eqn.int.v.x2}
	\\
	\int \left(\abs{\nabla f(x)}^2 + \frac{1}{2} v(x) f(x)^2\right)|x|^n \ud x 
		& \leq 
    \begin{cases}
    C a^{n+d-2} & n+d \leq 2d+1 \\
    C a^{n+d-2}\log b/a & n+d = 2d+2 \\
    C a^{2d} b^{n-d-2} & n+d \geq 2d+3 
    \end{cases}
	\label{eqn.int.v.xn}
	\\
	\abs{\int f(x) \abs{\nabla f(x)} |x|^n \ud x}
		& \leq \begin{cases}
		Ca^{d-1} & n=0 
		\\
    C a^d \log b/a & n=1
    \\
		C a^{d} b^{n-1} & n\geq 2
		\end{cases}
	\label{eqn.int.f.nabla.f.xn}
\end{align}
\endgroup
where $c_d$ is defined in \Cref{eqn.define.cd}.
\end{lemma}
\begin{proof}
\Cref{eqn.int.v.xn,eqn.int.v.x2,eqn.Ig.bdd,eqn.int.f.nabla.f.xn} all follow from the definition of the scattering length, \Cref{def.scattering.length},
and the bounds \cites[Lemma A.1]{Lieb.Yngvason.2001}[Lemma 2.2]{Lauritsen.Seiringer.2023}
\begin{equation*}
\left[1 - \frac{a^d}{|x|^d}\right]_+ \leq f_0(x) \leq 1,
\qquad 
\abs{\nabla f_0(x)} \leq \frac{da^d}{|x|^{d+1}} \quad \textnormal{for } |x|> a
\end{equation*}
where the left inequality in the first inequality is an equality for $x$ outside the support of $v$.
We refer to \cite[Equations (4.1) to (4.6)]{Lauritsen.Seiringer.2023} for a detailed proof. 
\end{proof}

\noindent
We will need the following technical lemma.
\begin{lemma}\label{lem.bdd.sum.gamma.k.ini}
Let $\hat \gamma(k) = z e^{-\beta|k|^2}$.
Let $p,n,m$ be non-negative integers with $1\leq n\leq m$. Then
\begin{equation*}
\begin{aligned}
	\frac{1}{L^d} \sum_{k\in \frac{2\pi}{L}\Z^d} \frac{|k|^p \hat \gamma(k)^n}{(1 + \hat \gamma(k))^m} 
	& = \frac{1}{(2\pi)^d} \int_{\R^d} \frac{|k|^p \hat \gamma(k)^n}{(1 + \hat\gamma(k))^m} \ud k 
		+ O\left(L^{-1} \beta \max \left\{\beta^{-1}, \mu \right\}^{\frac{p+d+1}{2}}\right)
	\\ &
	\leq 
  		C\max \left\{\beta^{-1}, \mu \right\}^{\frac{p+d}{2}}
\end{aligned}
\end{equation*}
for $z =e^{\beta \mu}  \gtrsim 1$ and $L$ sufficiently large.
\end{lemma}

\noindent
Note that $\hat \gamma(k) \ne \hat \gamma^{(1)}(k)$. In fact, $\hat\gamma^{(1)}(k) = \frac{\hat\gamma(k)}{1 + \hat\gamma(k)}$.

\begin{proof}
We interpret the sum as a Riemann sum and compare it with its corresponding integral
\begin{equation*}
	I(p,n,m) := \frac{1}{(2\pi)^d} \int_{\R^d} \frac{|k|^p \hat \gamma(k)^n}{(1 + \hat\gamma(k))^m} \ud k 
\end{equation*}
Writing $F_{p,n,m}(k) = \frac{|k|^p \hat \gamma(k)^n}{(1 + \hat \gamma(k))^m} $ then 
\begin{equation*}
\begin{aligned}
\frac{1}{L^d} \sum_{k\in \frac{2\pi}{L}\Z^d} \frac{|k|^p \hat \gamma(k)^n}{(1 + \hat \gamma(k))^m} 
	& = \frac{1}{(2\pi)^d} \sum_{k\in \frac{2\pi}{L}\Z^d} \int_{\left[-\frac{\pi}{L}, \frac{\pi}{L}\right]^d} 
		\left(F_{p,n,m}(k+\xi) - \int_0^1 \partial_t F_{p,n,m}(k+t\xi) \ud t\right) \ud \xi
\end{aligned}
\end{equation*}
The first term is the integral $I(p,n,m)$. 
For the second term we may bound 
\begin{equation*}
\abs{\partial_t F_{p,n,m}(k+t\xi)}
	\leq \begin{cases}
	C |\xi| \beta F_{p,n,m}(k+\xi) & p = 0
	\\
	C |\xi| \left[F_{p-1,n,m}(k+\xi) + \beta F_{p+1,n,m}(k+\xi)\right] & p \ne 0
	\end{cases}
\end{equation*}
Thus we have (with $I_{p-1,n,m} = 0$ if $p=0$)
\begin{equation*}
\begin{aligned}
\frac{1}{L^d} \sum_{k\in \frac{2\pi}{L}\Z^d} \frac{|k|^p \hat \gamma(k)^n}{(1 + \hat \gamma(k))^m} 
	& = I_{p,n,m} + O\left(L^{-1} I_{p-1,n,m} + L^{-1}\beta I_{p+1,n,m}\right)
\end{aligned}
\end{equation*}
We may bound the integrals $I_{p,n,m}$ as follows.
First, if $z \geq e$, i.e. $\beta\mu \geq 1$, we write
\begin{equation*}
\begin{aligned}
	\int_{\R^d} \frac{|k|^p \hat \gamma(k)^n}{(1 + \hat\gamma(k))^m} \ud k 
	& \lesssim \int_0^{\sqrt{\mu}} k^{p+d-1} \hat\gamma(k)^{n-m} \ud k + \int_{\sqrt{\mu}}^\infty k^{p+d-1} \hat\gamma(k)^n \ud k 
	\\
  & = \frac{1}{2} 
    \left[\int_{0}^{\beta\mu} \beta^{-1} \left(\frac{\beta\mu - t}{\beta}\right)^{\frac{p+d-2}{2}} e^{(n-m)t} \ud t
            + \int_0^\infty \beta^{-1} \left(\frac{\beta\mu + t}{\beta}\right)^{\frac{p+d-2}{2}} e^{-nt} \ud t\right]
  \\
  & \lesssim 
  \begin{cases}
  \mu^{\frac{p+d}{2}} & n=m, \\
  \mu^{\frac{p+d}{2}} (\beta\mu)^{-1} & n < m.
  \end{cases}
\end{aligned}
\end{equation*}
Next, if $z < e$ then 
\begin{equation*}
\begin{aligned}
  \int_{\R^d} \frac{|k|^p \hat \gamma(k)^n}{(1 + \hat\gamma(k))^m} \ud k 
  & \lesssim \int_{0}^\infty k^{p+d-1} \hat\gamma(k)^n \ud k 
  = \frac{1}{2} z^n \beta^{-\frac{p+d}{2}}
    \int_0^\infty t^{\frac{p+d-2}{2}} e^{-nt} \ud t
  \lesssim 
  \beta^{-\frac{p+d}{2}}.
\end{aligned}
\end{equation*}
The lemma follows.
\end{proof}

\noindent
Finally, we have the following lemma for the reduced densities of the free state.
\begin{lemma}\label{lem.prop.rho0}
The reduced densities of the free Fermi gas satisfy 
\begin{align}
\rho^{(1)}(x_1)
  & = \rho_0= \frac{1}{(4\pi)^{d/2}} \beta^{-d/2} (-\Li_{d/2}(-z))
    \left[1 + O(L^{-1}\zeta \rho_0^{-1/d})\right],
  \label{eqn.rho1.free}
  \\
  \rho^{(2)}(x_1,x_2)
  & = 2\pi \frac{- \Li_{d/2+1}(-z)}{(-\Li_{d/2}(-z))^{1+2/d}} \rho_0^{2+2/d} |x_1-x_2|^2
    \left[1 
    +  O(\rho_0^{2/d}|x_1-x_2|^2)
    +O (L^{-1} \zeta \rho_0^{-1/d} ) 
    \right].
  \label{eqn.rho2.free}
\end{align}
\end{lemma}	

\noindent
\Cref{eqn.rho2.free,eqn.rho1.free} are the finite volume analogues of \Cref{eqn.rho1.free.TL,eqn.rho2.free.TL}.
\begin{remark}\label{rmk.asym.beta}
Note that $\beta \sim \zeta \rho_0^{-2/d}$. (Recall that $\zeta = 1 + \abs{\log z}$.)
Indeed, for $z\leq C$ this is clear from \Cref{eqn.rho1.free}. For $z \gg 1$ 
this follows from the asymptotics of the polylogarithm, \Cref{eqn.asym.polylog}.
Moreover, for $\beta\mu \geq 1$ then $\mu \sim \rho_0^{2/d}$.
In particular then \Cref{lem.bdd.sum.gamma.k.ini} may be reformulated as 
\begin{equation}\label{lem.bdd.sum.gamma.k}
		\frac{1}{L^d} \sum_{k\in \frac{2\pi}{L}\Z^d} \frac{|k|^p \hat \gamma(k)^n}{(1 + \hat \gamma(k))^m} 
	= \frac{1}{L^d} \sum_{k\in \frac{2\pi}{L}\Z^d } \frac{|k|^p e^{- n \beta |k|^2}}{\left(1 + e^{-\beta|k|^2}\right)^m}
	\leq 
  C\rho_0^{1+p/d}
\end{equation}
for $z\gtrsim 1$ and $L$ sufficiently large.
This is the form we will later use.
\end{remark}

\begin{proof}
By translation invariance
\begin{equation*}
\begin{aligned}
\rho_0	
	& = \frac{\expect{\mcN}_0}{L^d} = \frac{1}{L^d} \int \rho^{(1)}(x) \ud x = \rho^{(1)}(0).
\end{aligned}
\end{equation*}
Moreover, by \Cref{lem.bdd.sum.gamma.k.ini}
\begin{equation*}
\begin{aligned}
\rho^{(1)}(0) 
	& = \frac{1}{L^d} \sum_{k\in \frac{2\pi}{L}\Z^d} \frac{e^{\beta\mu - \beta|k|^2}}{1 + e^{\beta\mu - \beta|k|^2}}
	\\
	& = \frac{1}{(2\pi)^d} \int_{\R^d} \frac{ze^{ - \beta|k|^2}}{1 + ze^{ - \beta|k|^2}} \ud k
		+ O\left(L^{-1}\beta \max \left\{\beta^{-1}, \mu\right\}^{\frac{d+1}{2}}\right)
	\\
	& = 
    \frac{\Gamma(d/2) \abs{\S^{d-1}}}{2(2\pi)^d} 
    \beta^{-d/2}
    (-\Li_{d/2}(-z))
		\left(1 + O\left(L^{-1}\beta \max \left\{\beta^{-1}, \mu\right\}^{1/2}\right)\right)
\end{aligned}
\end{equation*}
where $\abs{\S^{d-1}} = \frac{2\pi^{d/2}}{\Gamma(d/2)}$ is the surface area of the $(d-1)$-sphere.
Using that $\max\left\{\beta^{-1},\mu\right\} \sim \rho_0^{2/d}$ 
(which follow from this equation for $L$ sufficiently large, see \Cref{rmk.asym.beta})
we conclude the proof of \Cref{eqn.rho1.free}.

Next, we consider the $2$-particle density.
By Wick's rule we have 
\begin{equation*}
\rho^{(2)}(x_1,x_2)
	= \rho^{(1)}(x_1)\rho^{(1)}(x_2) - \gamma^{(1)}(x_1;x_2)\gamma^{(1)}(x_2;x_1).
\end{equation*}
By translation invariance $\gamma^{(1)}(x_1;x_2)$ is a function of $x_1-x_2$ only. 
We expand it as a Taylor series in $x_1-x_2$. 
By symmetry of reflection in any of the axes all odd orders and all off-diagonal second order terms vanish.
Additionally, all second order terms are equal by the symmetry of permutation of the axes. 
That is,
\begin{equation*}
\begin{aligned}
	\gamma^{(1)}(x_1;x_2)
		& = \frac{1}{L^d}\sum_{k} \hat \gamma^{(1)}(k) e^{ik(x_1-x_2)}
		\\
		& = \frac{1}{L^d}\sum_{k} \hat \gamma^{(1)}(k) \left[1 - \frac{1}{2d}|k|^2 |x_1-x_2|^2 + O(|k|^4 |x_1-x_2|^4)\right]
		\\
		& = \rho_0- \frac{1}{2d} \left[\frac{1}{L^d} \sum_{k} |k|^2 \hat \gamma^{(1)}(k)\right] |x_1-x_2|^2
			+ O\left(\left[\frac{1}{L^d} \sum_{k} |k|^4 \hat \gamma^{(1)}(k)\right] |x_1-x_2|^4\right).
\end{aligned}
\end{equation*}
(Here $O(|k|^4 |x_1-x_2|^4)$ means a term that is bounded by $|k|^4 |x_1-x_2|^4$ uniformly in $|k|^4 |x_1-x_2|^4$,
even if it is large.)
For the first sum we have by \Cref{lem.bdd.sum.gamma.k.ini,eqn.rho1.free} 
(and writing the error term in terms of $\rho_0$ as above) 
\begin{equation*}
\begin{aligned}
	\frac{1}{L^d} \sum_{k\in \frac{2\pi}{L}\Z^d} |k|^2 \hat\gamma^{(1)}(k)
	& = \frac{1}{(2\pi)^d} \int \frac{ze^{- \beta|k|^2}}{1 + ze^{- \beta|k|^2}} |k|^2 \ud k 
		+ O\left(L^{-1} \zeta \rho_0^{4/d}\right)
	\\
	& =
  \frac{\Gamma(d/2+1) \abs{\S^{d-1}}}{2(2\pi)^d} 
  \beta^{-d/2-1} (-\Li_{d/2+1}(-z))
		\left(1 + O(L^{-1}\zeta \rho_0^{-1/d})\right)
	\\
	& = 2d\pi \frac{- \Li_{d/2+1}(-z)}{(-\Li_{d/2}(-z))^{1+2/d}} \rho_0^{1+2/d}
		\left(1 + O(L^{-1}\zeta \rho_0^{-1/d})\right).
\end{aligned}
\end{equation*}
Using again \Cref{lem.bdd.sum.gamma.k.ini} to bound the second sum we conclude that 
\begin{equation*}
\begin{aligned}
	\gamma^{(1)}(x_1;x_2)
		& = \rho_0- \pi \frac{- \Li_{d/2+1}(-z)}{(-\Li_{d/2}(-z))^{1+2/d}} \rho_0^{1+2/d} |x_1-x_2|^2
		\\
		& \qquad 
			+ O\left(L^{-1} \zeta \rho_0^{1+1/d}|x_1-x_2|^2\right)
			+ O\left(\rho_0^{1+4/d}|x_1-x_2|^4\right).
\end{aligned}
\end{equation*}
We conclude the proof of \Cref{eqn.rho2.free}.
\end{proof}

\section{Gaudin--Gillepie--Ripka expansion}\label{sec.GGR.expansion}
We use the Gaudin--Gillepie--Ripka (GGR) expansion \cite{Gaudin.Gillespie.ea.1971} to compute $Z_J$ and $\rho^{(q)}_J$, the 
$q$-particle reduced densities of the trial state $\varGamma_J$.
For this we recall some notation from \cite{Lauritsen.Seiringer.2023}.
\begin{defn}[{\cite[Definition 3.1]{Lauritsen.Seiringer.2023}}]\label{defn.diagram}
We define $\mcG_p^{q}$ as the set of graphs on 
$q$ \emph{external} vertices $\{1,...,q\}$
and $p$ \emph{internal} vertices $\{q+1,...,q+p\}$ 
such that 
there are no edges between external vertices and such that all internal vertices have degree at least $1$,
i.e. there is at least one edge incident to each internal vertex.
We replace $q$ and/or $p$ with sets $V^*$ and $V$ respectively 
and write $\mcG_V^{V^*}$ if we need the external and/or internal vertices to have definite indices $V^*$ respectively $V$.


Define $\mcT_p^q\subset \mcC_p^q\subset \mcG_p^q$ as the subset of trees and connected graphs respectively. 
(Define similarly $\mcT_V^{V^*}\subset \mcC_V^{V^*} \subset \mcG_V^{V^*}$.)
Define the functions
\begin{equation*}
	W_p^q = W_p^q(x_1,\ldots,x_{p+q}) = \sum_{G\in\mcG_p^q} \prod_{e\in G} g_e. 
\end{equation*}
A \emph{diagram} $(\pi,G)$ (on $q$ external and $p$ internal vertices)
is a pair of a permutation $\pi\in \mcS_{p+q}$ and a graph $G\in \mcG_p^q$.
We view the permutation $\pi$ as a directed graph on the $p+q$ vertices.
The set of all diagrams on $q$ external and $p$ internal vertices is denoted $\mcD_p^q$.

For a diagram $(\pi,G)$ we will refer to $G$ as the $g$-graph and $\pi$ as the $\gamma$-graph.
The value of a diagram $(\pi,G) \in \mcD_p^q$ is the function 
\begin{equation*}
	\Gamma_{\pi,G}^q(x_1,\ldots,x_q) 
		= (-1)^{\pi} \idotsint \prod_{j=1}^{p+q} \gamma^{(1)}(x_j;x_{\pi(j)})
			\prod_{e\in G} g_e  \ud X_{[q+1,q+p]}.
\end{equation*}
A diagram $(\pi,G)\in\mcD_p^q$ is \emph{linked} if the union of $\pi$ and $G$ is a connected graph.
The subset of all linked diagrams is denoted $\mcL_p^q \subset \mcD_p^q$.

Define the set $\tilde\mcL_p^q \subset \mcD_p^q$ as the set of all diagrams such that each 
linked component contains at least one external vertex.
If $q=0$ we write $\mcG_p^q = \mcG_p$ etc. without a superscript $q$.
\end{defn}

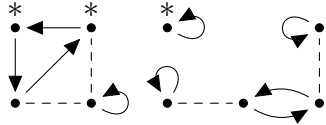
\begin{figure}[htb]
\centering
\begin{tikzpicture}[line cap=round,line join=round,>=triangle 45,x=1.0cm,y=1.0cm]
\foreach \x in {1,...,5}
\foreach \y in {3,4}
\node (\x\y) at (\x, \y) {};
\draw[dashed] (13) -- (23) -- (24);
\draw[dashed] (33) -- (43); 
\draw[dashed] (53) -- (54);
\draw[->] (14) to (13);
\draw[->] (13) to (24);
\draw[->] (24) to (14);
\draw[->] (23) to[out=-30,in=30,loop] ();
\draw[->] (33) to[out=60,in=120,loop] ();
\draw[->] (43) to[bend right] (53);
\draw[->] (53) to[bend right] (43);
\draw[->] (34) to[out=-30,in=30,loop] ();
\draw[->] (54) to[out=150,in=210,loop] ();
\foreach \i in {13,14,23,33,43,34,24,53,54} \draw[fill] (\i) circle [radius=1.5pt];
\foreach \i in {14,24,34} \node[anchor = south] at (\i) {$*$};
\end{tikzpicture}
\caption{Example of a diagram $(\pi,G)\in \mcD_{6}^3$ with $3$ linked component. Vertices labelled with $*$ denote external vertices, 
dashed lines denote $g$-edges and arrows denote $\gamma$-edges, i.e. an arrow from $i$ to $j$ denotes that $\pi(i) = j$.
Note that all internal vertices have at least one incident $g$-edge, 
that external vertices may have none, and that there are no $g$-edges between external vertices.
}
\label{fig.diagram}
\end{figure}

\begin{notation}
By a picture of a diagram, such as \Cref{fig.diagram}, 
we will also denote the value of the pictured diagram.
\end{notation}

To formulate our convergence criterion for the GGR expansion we additionally define the quantities
\begin{equation}\label{eqn.def.gamma.Ig.Igamma}
\begin{aligned}
	I_\gamma
		= \int_{[0,L]^d} \abs{\gamma^{(1)}(x)} \ud x,
	\qquad 
	I_g
		= \int_{\R^d} |g(x)| \ud x
\end{aligned}
\end{equation}
We shall bound $I_\gamma$ and $I_g$ in \Cref{lem.bdd.gamma.Ig.Igamma} below.
Note that $\frac{1}{L^d}\sum_{k\in \frac{2\pi}{L}\Z^d} \abs{\hat\gamma^{(1)}(k)} = \rho_0$.

\subsection{Calculation of \texorpdfstring{$Z_J$}{ZJ}}\label{sec.calc.ZJ}
We calculate $Z_{J}$. For simplicity  denote the diagonal of 
$\varGamma_n$ by $\varGamma_n = \varGamma_n(X_n) = \varGamma_n(X_n; X_n)$.
Then
\begin{equation*}
\begin{aligned}
Z_J
	& = \sum_{n=0}^{\infty} \idotsint \prod_{i<j} f_{ij}^2 \varGamma_n(X_n) \ud X_n
	\\
	& 
	= \sum_{n=0}^{\infty} \idotsint \prod_{i<j} (1+g_{ij}) \varGamma_n(X_n) \ud X_n 
	\\
	&
	= \sum_{n=0}^\infty \idotsint \left[1 + \sum_{p=2}^{n} \frac{n!}{(n-p)!p!}  W_p\right] \varGamma_n \ud X_n.
\end{aligned}
\end{equation*}
Now, if 
\begin{equation*}
	\sum_{n=0}^\infty \sum_{p=2}^{n} \frac{n!}{(n-p)!p!} \idotsint |W_p| \varGamma_n \ud X_n < \infty
\end{equation*}
then we may interchange the two sums. 
A criterion for this is given in  \Cref{lem.new.condition.abs.conv} below.
Thus, if the condition of \Cref{lem.new.condition.abs.conv} is satisfied,
namely that $\rho_0I_g$ is sufficiently small,
we have
\begin{equation*}
\begin{aligned}
Z_J
	& = Z \left[1 + 
		\sum_{p=2}^\infty \frac{1}{p!} 
		\idotsint \ud X_p W_p \left[\frac{1}{Z}\sum_{n=p}^{\infty} \frac{n!}{(n-p)!} \idotsint \ud X_{[p+1,n]}  
		\varGamma_n \right]\right]	
	\\
	& = Z \left[1+\sum_{p=2}^\infty \frac{1}{p!} \idotsint \ud X_p W_p \rho^{(p)}\right].
\end{aligned}
\end{equation*}
The free Fermi gas is a quasi-free state, and thus by Wick's rule we have
\begin{equation*}
	Z_J = 
		 Z \left[1 + \sum_{p=2}^\infty \frac{1}{p!} \idotsint \ud X_p W_p \det\left[\gamma^{(1)}_{ij}\right]_{i,j\leq p}\right]
\end{equation*}
Expanding $W_p$ and the determinant we get
\begin{equation*}
	Z_J = 
		Z \left[1+\sum_{p=2}^\infty \frac{1}{p!} \sum_{(\pi,G)\in \mcD_p} \Gamma_{\pi,G}\right].
\end{equation*}
This is exactly of the form where we can use the GGR expansion \cite[Lemma 3.6]{Lauritsen.2023}, \cite[Theorem 3.4]{Lauritsen.Seiringer.2023}.
From that we conclude that if $\rho_0I_g I_\gamma$ and $\rho_0I_g$ are sufficiently small then
\begin{equation}\label{eqn.Z_J.ini}
	Z_J = 
		Z \exp\left[\sum_{p=2}^\infty \frac{1}{p!} \sum_{(\pi,G)\in \mcL_p} \Gamma_{\pi,G}\right].
\end{equation}

\subsection{Calculation of \texorpdfstring{$\rho^{(q)}_J$}{rho(q)J}}\label{sec.calc.rhoJ}
Next, we calculate the reduced densities $\rho_J^{(q)}$ of the trial state $\varGamma_J$.
We have similarly as for $Z_J$ above
\begin{equation*}
\begin{aligned}
	\rho_J^{(q)}
	& = \frac{1}{Z_J} \sum_{n=q}^{\infty} \frac{n!}{(n-q)!}\idotsint \prod_{1\leq i < j \leq n} f_{ij}^2 \varGamma_n \ud X_{[q+1,n]}
	\\
	& = \frac{1}{Z_J} \prod_{1\leq i < j \leq q} f_{ij}^2 
		\sum_{n=q}^{\infty} \frac{n!}{(n-q)!} \sum_{p=0}^{n-q} \frac{(n-q)!}{p! (n-q-p)!}\idotsint W_p^q \varGamma_n \ud X_{[q+1,n]}.
\end{aligned}
\end{equation*}
By \Cref{lem.new.condition.abs.conv} below we may interchange the sums if $\rho_0I_g$ is sufficiently small.
Then 
\begin{equation*}
\begin{aligned}
	\rho_J^{(q)}
	& = 
		\frac{Z}{Z_J}  \prod_{1\leq i < j \leq q} f_{ij}^2  
			\sum_{p=0}^{\infty} \frac{1}{p!} \idotsint W_p^q 
			\left[\frac{1}{Z} \sum_{n=p+q}^\infty \frac{n!}{(n-p-q)!} \idotsint \varGamma_n \ud X_{[q+p+1,n]}\right]
			\ud X_{[q+1,q+p]}	
	\\
	& = \frac{Z}{Z_J} \prod_{1\leq i < j \leq q} f_{ij}^2  
		\sum_{p=0}^\infty \frac{1}{p!} \idotsint W_p^q \rho^{(p+q)}
			\ud X_{[q+1,q+p]}		 
\end{aligned}
\end{equation*}
Expanding the $W_p^q$ and using the Wick rule for the reduced densities of the free gas as above we get 
\begin{equation*}
\begin{aligned}
	\rho_J^{(q)}
	& = 
		\frac{Z}{Z_J}
		\prod_{1\leq i < j \leq q} f_{ij}^2
		\sum_{p=0}^\infty \frac{1}{p!} \sum_{(\pi,G)\in \mcD_p^q} \Gamma_{\pi,G}^q.
\end{aligned}
\end{equation*}
As above, we use the GGR expansion \cite[Lemma 3.6]{Lauritsen.2023}, \cite[Theorem 3.4]{Lauritsen.Seiringer.2023} to get 
\begin{equation}\label{eqn.rhoq_J.ini}
\begin{aligned}
	\rho_J^{(q)}
	& = 
		\prod_{1\leq i < j \leq q} f_{ij}^2  
		\sum_{p=0}^\infty \frac{1}{p!} \sum_{(\pi,G)\in \tilde\mcL_p^q} \Gamma_{\pi,G}^q
\end{aligned}
\end{equation}
for $\rho_0I_g I_\gamma$ and $\rho_0I_g$ small enough (dependent on $q$).

\subsection{A convergence criterion}
In this section we show
\begin{lemma}\label{lem.new.condition.abs.conv}
There exists a constant $c > 0$ such that if $\rho_0I_g < c$ then 
\begin{align}
\frac{1}{Z}
\sum_{n=0}^\infty \sum_{p=2}^{n} \frac{n!}{(n-p)!p!} \idotsint |W_p| |\varGamma_n| \ud X_n
	& 
	\leq \exp( CL^d \rho_0 I_g ) 
	< \infty,
\label{eqn.new.condition.q.=.0}
\\
\intertext{and for any $q \geq 1$}
\frac{1}{Z}
\sum_{n=q}^\infty \sum_{p=0}^{n-q} \frac{n!}{(n-q-p)!p!} \idotsint |W_p^q| |\varGamma_n| \ud X_{[q+1,n]}
	& 
	\leq C_q \rho_0^q \exp(C L^d \rho_0 I_g) 
	< \infty
\label{eqn.new.condition.q.geq.1}
\end{align}
uniformly in $x_1,\ldots,x_q$.
\end{lemma}

\begin{proof}
Write
\begin{equation*}
\begin{aligned}
\frac{1}{Z}
\sum_{n=0}^\infty \sum_{p=2}^{n} \frac{n!}{(n-p)!p!} \idotsint |W_p| |\varGamma_n| \ud X_n
	& = \frac{1}{Z} \sum_{p=2}^\infty\sum_{n=p}^{\infty} \frac{n!}{(n-p)!p!}\idotsint \ud X_n |W_p| \varGamma_n
	\\
	& = \sum_{p=2}^\infty \frac{1}{p!} \idotsint \ud X_p |W_p| \rho^{(p)}.
\end{aligned}
\end{equation*}
By splitting all graphs into their connected components we have
\begin{equation*}
W_p  = \sum_{G\in \mcG_p} \prod_{e\in G} g_e
	= \sum_{k=1}^\infty \frac{1}{k!} \sum_{n_1,\ldots,n_k\geq 2} \binom{p}{n_1,\ldots, n_k} \chi_{(\sum n_\ell = p)} 
		\prod_{\ell=1}^k \left[\sum_{G_\ell \in \mcC_{n_\ell}} \prod_{e_\in G_\ell} g_e\right].
\end{equation*}
Here $k$ is the number of connected components having sizes $n_1,\ldots,n_k$.
Note that $n_\ell \geq 2$ since each connected component needs at least $2$ vertices 
since any vertex in a graph $G\in \mcG_p$ is internal and hence connected to at least one other vertex.
The factor $\frac{1}{k!}$ comes from counting the possible labelings of the connected component 
and the factor $\binom{p}{n_1 \cdots n_k}$ comes from counting the possible labelings of the vertices in the different connected components.

By the tree-graph inequality \cite{Ueltschi.2018} we have
\begin{equation*}
\begin{aligned}
	\frac{1}{p!} \abs{W_p}
		& \leq \sum_{k=1}^\infty \frac{1}{k!} \sum_{n_1,\ldots,n_k\geq2} \frac{1}{n_1! \cdots n_k!} 
			\chi_{(\sum n_\ell = p)} \prod_{\ell=1}^k \left[\sum_{T_\ell \in \mcT_\ell} \prod_{e\in T_\ell} |g_e|\right]
\end{aligned}
\end{equation*}
By the Wick rule and Hadamard's inequality $\rho^{(p)} = \det \left[\gamma^{(1)}_{ij}\right]_{1\leq i,j\leq p} \leq \rho_0^p$.
Thus
\begin{align*}
	& \sum_{p=2}^\infty \frac{1}{p!} \idotsint \ud X_p |W_p| 
	\rho^{(p)}
	\\
	& \quad \leq 
		\sum_{k=1}^\infty \frac{1}{k! } 
		\sum_{n_1,\ldots,n_k\geq 2} \frac{1}{n_1!\cdots n_k!} 
		\rho_0^{\sum n_\ell}
		\prod_{\ell=1}^k\left[\sum_{T_\ell \in \mcT_\ell} \idotsint \prod_{e\in T_\ell} |g_e|\right]
\intertext{
For each tree the integration is over all variables, thus by the translation invariance the integration 
over the variables in the tree $T_\ell$ gives $L^d (\int |g|)^{n_\ell - 1}$.
Using moreover Cayley's formula $\# \mcT_n = n^{n-2}\leq C^n n!$ we get
}
	& \quad \leq 
		\sum_{k=1}^\infty \frac{1}{k!} 
		\sum_{n_1,\ldots,n_k\geq 2} \frac{1}{n_1!\cdots n_k!} 
		\rho_0^{\sum n_\ell}
		C^{\sum n_\ell} n_1! \cdots n_k! \left(\int |g|\right)^{\sum n_\ell - k} L^{dk}
	\\
	& \quad = 
		\sum_{k=1}^\infty \frac{1}{k!} 
		\left[C\rho_0L^d \sum_{n=2}^\infty (C\rho_0I_g)^{n-1} \right]^k 
	\\ & \quad 
		\leq \exp\left(CL^d \rho_0I_g\right)
		< \infty
\end{align*}
if $\rho_0I_g$ is sufficiently small.

The proof of \Cref{eqn.new.condition.q.geq.1} is in spirit the same. 
Write 
\begin{equation*}
\begin{aligned}
\frac{1}{Z}
\sum_{n=q}^\infty \sum_{p=0}^{n-q} \frac{n!}{(n-q-p)!p!} \idotsint |W_p^q| |\varGamma_n| \ud X_{[q+1,n]}
	& = \sum_{p=0}^\infty \frac{1}{p!} \idotsint \ud X_{[q+1,q+p]} |W_p^q| \rho^{(q+p)}.
\end{aligned}
\end{equation*}
By decomposing the graphs into their connected components we have 
\begin{equation*}
\begin{aligned}
W_p^q 
	& = 
	\sum_{\kappa=1}^q \frac{1}{\kappa!} 
	\sum_{\substack{(V_1^*,\ldots,V_\kappa^*) \\ \textnormal{partition of } \{1,\ldots,q\} \\ V_\lambda^* \ne \varnothing}}
	\sum_{n_1^*,\ldots,n_\kappa^* \geq 0}
	\sum_{k=0}^\infty \frac{1}{k!}
	\sum_{n_1,\ldots,n_k\geq 2}
	\binom{p}{n_1^*,\ldots,n_\kappa^*, n_1,\ldots,n_k} \chi_{(\sum_\lambda n^*_\lambda + \sum_\ell n_\ell = p)}
	\\
	& \quad \times 
	\prod_{\lambda=1}^\kappa 
	\left[
		\sum_{G^*_\lambda \in \mcC_{n^*_\lambda}^{V_\lambda^*}} 
		\prod_{e\in G_\lambda^*} g_e
	\right]
	\prod_{\ell=1}^k 
	\left[
		\sum_{G_\ell \in \mcC_{n_\ell}}
		\prod_{e\in G_\ell} g_e
	\right].
\end{aligned}
\end{equation*}
Here $\kappa$ is the number of connected components having external vertices and $k$ is the number of connected components
only with internal vertices. 
The partition $(V_1^*,\ldots,V_\kappa^*)$ partitions the external vertices into the $\kappa$ different connected components
with external vertices and the numbers $n_1^*,\ldots,n_\kappa^*$ 
are the number of internal vertices in the connected components with external vertices. 
The numbers $n_1,\ldots,n_k$ and the combinatorial factors are as above.

Using the tree-graph inequality as above we will obtain a sum of trees. 
(Technically we need to use a trivial modification of the tree-graph bound adapted to the setting with external vertices,
see \cites[Section 3.1.3]{Lauritsen.Seiringer.2023}[Section 4.2]{Lauritsen.2023} for details.)
Namely we will have factors like
\begin{equation*}
\sum_{T_\lambda^* \in \mcT_{n_\lambda^*}^{V^*_\lambda}} \prod_{e \in T_\lambda^*} |g_e|.
\end{equation*}
We bound these as follows. 
If $\# V_\lambda^* = 1$ we do nothing and define $T_{\lambda,1}^* = T_\lambda^*$.
Otherwise iteratively pick any edge on the path between any two external vertices
and bound the factor $|g_e|\leq 1$. Remove this edge from $T_\lambda^*$. 
Repeating this procedure $\# V^*_\lambda - 1$ many times results in $\# V^*_\lambda$ many trees all with exactly $1$ external vertex.
Label these as $T_{\lambda,1}^*,\ldots,T_{\lambda,\#V_\lambda^*}^*$.
We then have the bound 
\begin{equation*}
\prod_{e \in T_\lambda^*} |g_e|
	\leq \prod_{\nu = 1}^{\#V_\lambda^*} \prod_{e\in T_{\lambda,\nu}^*} |g_e|.
\end{equation*}
Using this bound together with Hadamard's inequality 
as above we get
\begin{align*}
	& \sum_{p=0}^\infty \frac{1}{p!} \idotsint \ud X_{[q+1,q+p]} |W_p^q| \rho^{(q+p)}
	\\
	& \quad 
		\leq 
		\sum_{\kappa=1}^q \frac{1}{\kappa!}
		\sum_{\substack{(V_1^*,\ldots,V_\kappa^*) \\ \textnormal{part. of } \{1,\ldots,q\} \\ V_\lambda^* \ne \varnothing}}
		\sum_{n_1^*,\ldots,n_\kappa^* \geq 0}
		\sum_{k=0}^\infty \frac{1}{k!}
		\sum_{n_1,\ldots,n_k\geq 2}
		\frac{1}{\prod_{\lambda=1}^\kappa n_\lambda^* ! \prod_{\ell=1}^k n_\ell !}
		\rho_0^{q + \sum_\lambda n_\lambda^* + \sum_\ell n_\ell}
	\\
	& \qquad \times
		\sum_{T_\lambda^* \in \mcT_{n_\lambda^*}^{\# V_\lambda^*}} 
		\sum_{T_\ell \in \mcT_{n_\ell}}
		\left[
		\prod_{\lambda=1}^\kappa
		\prod_{\nu = 1}^{\#V_\lambda^*} 
		\idotsint
		\prod_{e\in T_{\lambda,\nu}^*} |g_e|
		\right]
		\left[
		\prod_{\ell=1}^k 
		\idotsint
		\prod_{e\in T_{\ell}}
		|g_e|
		\right].
\intertext{
In the integrations each tree $T_{\lambda,\nu}^*$ is integrated over all but the one external vertex and so gives a value 
$(\int |g|)^{\# T_{\lambda,\nu}^* - 1}$
and each tree $T_\ell$ is integrated over all coordinates giving the value $(\int |g|)^{n_\ell - 1} L^d$.
Moreover $\sum_\nu (\# T_{\lambda,\nu}^* - 1) = n_\lambda^*$.
Thus, using additionally Cayley's formula,
}
	& \quad 
		\leq 
		\sum_{\kappa=1}^q \frac{1}{\kappa!}
		\sum_{\substack{(V_1^*,\ldots,V_\kappa^*) \\ \textnormal{part. of } \{1,\ldots,q\} \\ V_\lambda^* \ne \varnothing}}
		\sum_{n_1^*,\ldots,n_\kappa^* \geq 0}
		\sum_{k=0}^\infty \frac{1}{k!}
		\sum_{n_1,\ldots,n_k\geq 2}
		\frac{1}{\prod_{\lambda=1}^\kappa n_\lambda^* ! \prod_{\ell=1}^k n_\ell !}
		\rho_0^{q + \sum_\lambda n_\lambda^* + \sum_\ell n_\ell}
	\\
	& \qquad \times 
		C^{\sum_\lambda n_\lambda^* + \sum_\ell n_\ell}
		\prod_{\lambda=1}^\kappa (n_\lambda^* + \# V_\lambda^*)! 
		\prod_{\ell=1}^k n_\ell!
		\left(\int |g|\right)^{\sum_\lambda n_\lambda^* + \sum_\ell n_\ell}
		L^{dk}.
\intertext{
Next, we may bound the binomial coefficients as $(n+m)! \leq 2^{n+m}n!m!$ so 
$\prod_{\lambda=1}^\kappa (n_\lambda^* + \# V_\lambda^*)! \leq 2^{\sum_\lambda n_\lambda^* + q} \prod_\lambda n_\lambda^* ! (\# V_\lambda^*)!$. 
Thus
}
	& \quad 
		\leq 
			2^q
			\sum_{\kappa=1}^q \frac{1}{\kappa!} 
			\sum_{\substack{(V_1^*,\ldots,V_\kappa^*) \\ \textnormal{part. of } \{1,\ldots,q\} \\ V_\lambda^* \ne \varnothing}}
			\prod_{\lambda=1}^\kappa (\# V_\lambda^*)!
		\left[
			\sum_{n^* = 0}^\infty (C \rho_0I_g)^{n_*}
		\right]^{\kappa}
		\rho_0^q
		\sum_{k=0}^{\infty}
		\frac{1}{k!}
		\left[
			CL^d \rho_0\sum_{n=2}^\infty \left(C\rho_0I_g\right)^{n-1}
		\right]^k
	\\
	& \quad 
		\leq 
		C_q \rho_0^q \exp\left(CL^d \rho_0I_g\right) < \infty
\end{align*}
if $\rho I_g$ is small enough.
\end{proof}

\subsection{Calculation of \texorpdfstring{$I_g, I_\gamma$}{Ig, Igamma}}\label{sec.bdd.gamma.Ig.Igamma}
In this section we bound the quantities 
$I_g$ and $I_\gamma$ defined in \Cref{eqn.def.gamma.Ig.Igamma} above.
We show (recall $\zeta = 1 + \abs{\log z}$)
\begin{lemma}\label{lem.bdd.gamma.Ig.Igamma}
The quantities $I_g$ and $I_\gamma$ 
satisfy 
\begin{equation*}
\begin{aligned}
	I_g 
	\lesssim a^d \log b/a,
	\qquad
	I_\gamma 
	\lesssim \zeta^{d/2}.
\end{aligned}
\end{equation*}
\end{lemma}

\noindent
Note that these bounds are uniform in the volume $L^d$.


\begin{proof}
The bound $I_g \leq C a^d \log b/a$ follows from \Cref{eqn.Ig.bdd}.
For $I_\gamma$ we have for any (length) $\lambda > 0$
\begin{equation*}
\begin{aligned}
	I_\gamma
	& 
  = \int_{[0,L]^d} \abs{\gamma^{(1)}(x)}
	\\ 	& 
	\leq \left(\int_{\R^d} \abs{\gamma^{(1)}(x)}^2 (\lambda^2 + |x|^2)^{2} \ud x\right)^{1/2} \left(\int_{\R^d} \frac{1}{(\lambda^2+|x|^2)^2} \ud x \right)^{1/2}	
	\\ 	& 
  = C \lambda^{d/2-2} \left[\frac{1}{L^d} \sum_{k\in \frac{2\pi}{L}\Z^d} \abs{\widehat{(\lambda^2+|x|^2)\gamma^{(1)}}(k)}^2\right]^{1/2}
\end{aligned}
\end{equation*}
Moreover (with $\hat\gamma(k) = ze^{- \beta|k|^2}$ is as in \Cref{lem.bdd.sum.gamma.k})
\begin{equation*}
\begin{aligned}
\widehat{(\lambda^2+|x|^2)\gamma^{(1)}}(k)
	& = \left[\lambda^2 - \Delta_k\right] \hat \gamma^{(1)}(k)
	\\
	& = \frac{\lambda^2 \hat\gamma(k)^3 + (2\lambda^2 - 4\beta^2|k|^2 - 2d\beta) \hat \gamma(k)^2 
		+ (\lambda^2 + 4\beta^2|k|^2 - 2d\beta) \hat\gamma(k)
			}{(1 + \hat\gamma(k))^3}.
\end{aligned}
\end{equation*}
Using \Cref{lem.bdd.sum.gamma.k} we conclude that 
\begin{equation*}
\frac{1}{L^d} \sum_{k\in\frac{2\pi}{L}\Z^d} \abs{\widehat{(\lambda^2+|x|^2)\gamma^{(1)}}(k)}^2
  \leq C \rho_0\left(\lambda^4 + \beta^{4}\rho_0^{4/d} + \beta^2\right)
  \leq C \rho_0\left(\lambda^4 + \zeta^2\beta^2\right).
\end{equation*}
Thus for $\lambda = \beta^{1/2}\zeta^{1/2}$ we have $I_\gamma\lesssim \zeta^{d/2}$.
(Recall that $\beta \sim \zeta \rho_0^{-2/d}$ by \Cref{rmk.asym.beta}.)
\end{proof}

\noindent
We conclude that $\rho_0I_g I_\gamma \leq C a^d \rho_0\zeta^{d/2} \log b/a$.
Thus, the conditions of the calculations in \Cref{sec.calc.ZJ,sec.calc.rhoJ} are valid for $a^d\rho_0\zeta^{d/2}  \log b/a$ small enough,
independently of the volume $L^d$.
Recalling \Cref{eqn.Z_J.ini,eqn.rhoq_J.ini} we summarize the computations of this section.

\begin{lemma}\label{lem.conv.GGR}
For any $q_0$ there exists a constant $c_{q_0} > 0$ independently of $L$ such that if $a^d\rho_0\zeta^{d/2}  \log b/a < c_{q_0}$ then 
\begin{align}
	Z_J 
	& = 
		Z \exp\left[\sum_{p=2}^\infty \frac{1}{p!} \sum_{(\pi,G)\in \mcL_p} \Gamma_{\pi,G}\right],
	\label{eqn.Z_J}
	\\
	\rho_J^{(q)}
	& = 
		\prod_{1\leq i < j \leq q} f_{ij}^2  
		\sum_{p=0}^\infty \frac{1}{p!} \sum_{(\pi,G)\in \tilde\mcL_p^q} \Gamma_{\pi,G}^q.
	\label{eqn.rhoq_J}
\end{align}
for any $q \leq q_0$.
\end{lemma}

\section{Calculation of terms}\label{sec.calc.terms}
In this section we compute and bound the different terms in \Cref{eqn.calc.free.energy.initial}
and thereby prove \Cref{lem.pressure.high.temp}.

\subsection{Energy}
The kinetic energy of the trial state $\varGamma_J$ is 
\begin{equation*}
\begin{aligned}
	\expect{\mcH}_J
	& = \frac{1}{Z_J} \sum_{n=1}^\infty \idotsint \left[(-\Delta_{X_n}) \left[F(X_n) \varGamma(X_n, Y_n) F(Y_n)\right] \right]_{Y_n = X_n} \ud X_n
	\\
	& = \frac{1}{Z_J} \sum_{n=1}^\infty 
		\idotsint \left(|\nabla_{X_n} F|^2 \varGamma_n(X_n;X_n) - F_n^2 (\Delta_{X_n} \varGamma_n)(X_n;X_n)\right) \ud X_n.
\end{aligned}
\end{equation*}
The second term may be calculated as 
\begin{equation*}
\begin{aligned}
\frac{1}{Z_J} \sum_{n=1}^\infty \idotsint F_n^2 (-\Delta_{X_n} \varGamma_n)(X_n;X_n) \ud X_n
	& = \frac{Z}{Z_J} \Tr[F^2 \mcH \varGamma]
	\\
	& = \frac{1}{Z_J} \Tr[ F^2 (-\partial_\beta (Z\varGamma) + \mu\mcN Z \varGamma)]
	\\
	& = - \partial_\beta \log Z_J + \mu \expect{\mcN}_J
\end{aligned}
\end{equation*}
For the first term we have that 
\begin{equation*}
|\nabla_{X_n} F_n|^2
	= \left[2\sum_{j<k} \abs{\frac{\nabla f_{jk}}{f_{jk}}}^2
				+ \sum_{\substack{i,j,k \\ \textnormal{all distinct}}} \frac{\nabla f_{ij}\nabla f_{jk}}{f_{ij} f_{jk}}\right] F_n^2.
\end{equation*}
Thus the full energy is 
\begin{equation}\label{eqn.energy}
\begin{aligned}
\expect{\mcH  - \mu\mcN + \mcV}_J
	& =  - \partial_\beta \log Z_J + \iint \left[\abs{\frac{\nabla f_{12}}{f_{12}}}^2 + \frac{1}{2} v_{12}\right] \rho^{(2)}_J \ud x_1 \ud x_2
	\\ & \quad 
		+ \iiint \frac{\nabla f_{12} \nabla f_{13}}{f_{12} f_{13}}\rho^{(3)}_J \ud x_1 \ud x_2 \ud x_3.
\end{aligned}
\end{equation}


\subsection{Entropy}
We note that $\varGamma_J = \frac{Z}{Z_J} F \varGamma F$ is isospectral to $\frac{Z}{Z_J} \varGamma^{1/2} F^2 \varGamma^{1/2}$. 
Moreover, since $F\leq 1$ we have $\varGamma^{1/2} F^2 \varGamma^{1/2}\leq \varGamma$ as operators. 
Thus by operator monotonicity of the logarithm
\begin{equation*}
\begin{aligned}
	 \Tr\left[\varGamma_J \log \varGamma_J\right]
	& =   \frac{Z}{Z_J} \Tr\left[\varGamma^{1/2} F^2 \varGamma^{1/2} \left(\log \frac{Z}{Z_J} + \log   \varGamma^{1/2} F^2 \varGamma^{1/2}\right)\right]
	\\
	& \leq \log \frac{Z}{Z_J} + \frac{Z}{Z_J} \Tr\left[\varGamma^{1/2} F^2 \varGamma^{1/2}\log  \varGamma\right]
	\\
	& =  - \log Z_J - \beta \frac{Z}{Z_J} \Tr\left[F^2 \varGamma (\mcH - \mu \mcN)\right]
	\\
	& = - \log Z_J + \frac{1}{Z_J} \beta \partial_\beta \Tr\left[F^2 Z \varGamma\right]
	\\
	& = - \log Z_J + \beta \partial_\beta \log Z_J
\end{aligned}
\end{equation*}
We conclude the bound on the entropy
\begin{equation}\label{eqn.entropy}
\begin{aligned}
	-\frac{1}{\beta} S(\varGamma_J)
	& = \frac{1}{\beta} \Tr[\varGamma_J \log \varGamma_J]
	\leq -\frac{1}{\beta} \log Z_J + \partial_\beta \log Z_J.
\end{aligned}
\end{equation}

\subsection{Pressure}
Combining \Cref{eqn.energy,eqn.entropy} the terms $\pm \partial_\beta \log Z_j$ cancel and we conclude the bound for the pressure
\begin{equation*}
\begin{aligned}
  L^d P[\varGamma_J]
	&
	= - \expect{\mcH  - \mu\mcN + \mcV}_J + \frac{1}{\beta} S(\varGamma_J)
	\\ & 
	\geq 
		 \frac{1}{\beta} \log Z_J
		- \iint \left[\abs{\frac{\nabla f_{12}}{f_{12}}}^2 + \frac{1}{2} v_{12}\right] \rho^{(2)}_J \ud x_1 \ud x_2
		- \iiint \frac{\nabla f_{12} \nabla f_{13}}{f_{12} f_{13}}\rho^{(3)}_J \ud x_1 \ud x_2 \ud x_3.
\end{aligned}
\end{equation*}

\begin{remark}\label{rmk.energy.vs.free}
The cancellation of the terms $\pm \partial_\beta \log Z_j$ is not essential. 
Namely the energy of the trial state $\varGamma_J$ is the energy of the free gas plus the relevant interaction term up to small errors.
And the entropy of the trial state $\varGamma_J$ is bounded from above by the entropy of the free gas up to small errors.
To see this write
\[
    -\partial_\beta \log Z_J = -\partial_\beta \log Z - \partial_\beta \log \frac{Z_J}{Z}
        = \expect{\mcH - \mu\mcN}_0 - \partial_\beta \log \frac{Z_J}{Z}.
\]
One can show that $\partial_\beta \log \frac{Z_J}{Z}$ is small compared to the interaction of order $L^d a^d \rho_0^{2+2/d}$.
Thus, the energy of the trial state $\varGamma_J$ is 
\[
    \expect{\mcH  - \mu\mcN + \mcV}_J
    = \expect{\mcH - \mu\mcN}_0 + \iint \left[\abs{\frac{\nabla f_{12}}{f_{12}}}^2 + \frac{1}{2} v_{12}\right] \rho^{(2)}_J \ud x_1 \ud x_2
    + \textnormal{small error}.
\]
Similarly for the entropy
\begin{equation*}
\begin{aligned}
-\frac{1}{\beta} \log Z_J + \partial_\beta \log Z_J
  & = -\frac{1}{\beta}\log Z + \partial_\beta \log Z - \frac{1}{\beta} \log \frac{Z_J}{Z} + \partial_\beta \log \frac{Z_J}{Z}
  \\ 
  & = -\frac{1}{\beta}S(\varGamma) - \frac{1}{\beta} \log \frac{Z_J}{Z} + \partial_\beta \log \frac{Z_J}{Z}.
\end{aligned}
\end{equation*}
We show below that $\frac{1}{\beta} \log \frac{Z_J}{Z}$ is small compared to the interaction term of size $L^d a^d \rho_0^{2+2/d}$.
Thus the entropy of the trial state $\varGamma_J$ may be bounded as
\[
  -\frac{1}{\beta} S(\varGamma_J)
  \leq -\frac{1}{\beta} S(\varGamma) + \textnormal{small error}.
\]
The proof that $\partial_\beta \log \frac{Z_J}{Z}$ is small is somewhat analogous to the proof of \Cref{lem.bdd.eps.errors} in \Cref{sec.bdd.eps.errors}.
As we will not need it, we omit the details.
\end{remark}

By \Cref{eqn.rhoq_J}, we have for $a^d\rho_0\zeta^{d/2} \log b/a$ sufficiently small that 
\begin{equation*}
\rho_J^{(2)} = f_{12}^2 \left[\rho^{(2)} + \sum_{p=1}^\infty \frac{1}{p!} \sum_{(\pi,G)\in \tilde\mcL_p^2} \Gamma_{\pi,G}^2\right]
\end{equation*}
We may then write 
\begin{equation}\label{eqn.free.energy}
\begin{aligned}
  & L^d P[\varGamma_J]
  \\ & \quad \geq
		\frac{1}{\beta} \log Z
		- \iint \left[\abs{\nabla f_{12}}^2 + \frac{1}{2} v_{12} f_{12}^2\right] \rho^{(2)} \ud x_1 \ud x_2
		+ \underbrace{\frac{1}{\beta} \log \frac{Z_J}{Z}}_{\eps_Z}
	\\ & \qquad 
		- \underbrace{\iint \left[\abs{\nabla f_{12}}^2 + \frac{1}{2} v_{12}f_{12}^2\right] 
			\sum_{p=1}^\infty \frac{1}{p!} \sum_{(\pi,G)\in \tilde\mcL_p^2} \Gamma_{\pi,G}^2 \ud x_1 \ud x_2}_{\eps_{2}}
		- \underbrace{
		\vphantom{\iint \left[\abs{\nabla f_{12}}^2 + \frac{1}{2} v_{12}f_{12}^2\right] 
			\sum_{p=1}^\infty \frac{1}{p!} \sum_{(\pi,G)\in \tilde\mcL_p^2} \Gamma_{\pi,G}^2 \ud x_1 \ud x_2}
		\iiint \frac{\nabla f_{12} \nabla f_{13}}{f_{12} f_{13}}\rho^{(3)}_J \ud x_1 \ud x_2 \ud x_3}_{\eps_{3}}.
\end{aligned}
\end{equation}
The first term is the pressure of the free gas (times the volume), 
the second term leads to the leading order correction, 
and the remaining terms are error terms.
We shall show in \Cref{sec.bdd.eps.errors} below the following bounds. (Recall that $\zeta = 1 + \abs{\log z}$.)
\begin{lemma}\label{lem.bdd.eps.errors}
For $z\gtrsim 1$ there exists a constant $c > 0$ such that if $a^d\rho_0 \zeta^{d/2} \abs{\log a^d \rho_0} < c$
then, for sufficently large $L$, the error-terms are bounded as 
\begin{align*}
	\frac{\abs{\eps_Z}}{L^d} 
  & \lesssim 
     a^db^2\rho_0^{2+4/d}\zeta^{-1} 
    + a^{2d} \rho_0^{3 + 2/d} \zeta^{d/2-1}(\log b/a)^2
	\\
	\frac{\abs{\eps_2}}{L^d} 
  & \lesssim 
  \begin{cases}
  a^{2d} \rho_0^{3+2/d} \log b/a
  +
  a^{4d-2}\rho_0^5 \zeta^{3d/2}  (\log b/a)^3
  & d \geq 2,
  \\
  a b \rho_0^5 \log b/a 
  + a^2 \rho_0^5 \zeta^{3/2} (\log b/a)^3 
  & d = 1,
  \end{cases}
	\\
	\frac{\abs{\eps_3}}{L^d} & \lesssim
  \begin{cases}
     a^{2d} b^2 \rho_0^{3+4/d}
    + a^{3d-2}  \rho_0^4\zeta^{d/2} \log b/a
    & d\geq 2 ,
    \\
     a^2 \rho_0^5 \zeta   (\log b/a)^2 & d=1.
  \end{cases}
\end{align*}
\end{lemma}	

In particular we have the bounds (recalling that $a \ll b\lesssim \rho_0^{-1/d}$)
\begin{equation}\label{eqn.bdd.eps.terms.total}
    \frac{\abs{\eps_Z} + \abs{\eps_2} + \abs{\eps_3}}{L^d}
    \lesssim \begin{cases}
    a^3 b^2  \rho_0^{10/3}\zeta^{-1}
    + a^6 \rho_0^{11/3}\zeta^{1/2}  (\log b/a)^2
    + a^{10} \rho_0^5  \zeta^{9/2}(\log b/a)^3
    & d= 3,
    \\
    a^2 b^2 \rho_0^{4} \zeta^{-1}
    + a^4 \rho_0^{4}\zeta  \log b/a
    + a^{6}\rho_0^5 \zeta^{3}  (\log b/a)^3
    & d= 2,
    \\
    ab \rho_0^5 \log b/a 
    + a^2 \rho_0^5\zeta^{3/2}  (\log b/a)^3
    & d=1.
    \end{cases}
\end{equation}

\noindent
For the second term in \Cref{eqn.free.energy} above we use \Cref{eqn.rho2.free,eqn.int.v.x2,eqn.int.v.xn}, thus
\begin{equation}\label{eqn.calc.2.body.leading}
\begin{aligned}
	& \iint \left[\abs{\nabla f_{12}}^2 + \frac{1}{2} v_{12} f_{12}^2\right] \rho^{(2)} \ud x_1 \ud x_2
	\\ & \quad 
		= 2\pi \frac{-\Li_{d/2+1}(-e^{\beta\mu})}{(-\Li_{d/2}(-e^{\beta\mu}))^{1+2/d}} \rho_0^{2+2/d} 
		L^d \int \left(\abs{\nabla f}^2 + \frac{1}{2}v f^2\right) |x|^2 \ud x
		\left(1 + O(L^{-1}\zeta \rho_0^{-1/d})\right)
	\\
	& \qquad 
		+ O\left(L^d \rho_0^{2+4/d} \int \left(\abs{\nabla f}^2 + \frac{1}{2}v f^2\right) |x|^4 \ud x\right)
	\\
	& \quad =
		2\pi c_d \frac{-\Li_{d/2+1}(-e^{\beta\mu})}{(-\Li_{d/2}(-e^{\beta\mu}))^{1+2/d}} L^d a^d \rho_0^{2+2/d} 
			\left(1 + O(a^d/b^d) +  O\left(L^{-1} \zeta \rho_0^{-1/d}\right)\right)
  \\ & \qquad 
			+ 
      \begin{cases}
      O\left(L^d a^{d+2} \rho_0^{2 + 4/d} \log b/a\right) & d \geq 2
      \\
      O\left(L a^2 b \rho_0^{6}\right) & d = 1
      \end{cases}
\end{aligned}
\end{equation}
where $c_d$ is defined in \Cref{eqn.define.cd}.
Combining  \Cref{eqn.free.energy,eqn.calc.2.body.leading,eqn.bdd.eps.terms.total}
we thus conclude the bound 
\begin{equation*}
\begin{aligned}
  & \psi(\beta,\mu)
	\\ & \quad  
  \geq \limsup_{L\to \infty} P[\varGamma_J]
	\\
	& \quad 
  \geq \lim_{L\to \infty} \left[\frac{1}{L^d\beta} \log Z\right]
		- 2\pi c_d \frac{-\Li_{d/2+1}(-e^{\beta\mu})}{(-\Li_{d/2}(-e^{\beta\mu}))^{1+2/d}} a^d \rho_0^{2+2/d} 
	\\
	& \qquad 
    + \begin{cases}
    O
    \begin{pmatrix}
    a^{6}b^{-3} \rho_0^{8/3}
    + a^3 b^2  \rho_0^{10/3}\zeta^{-1}
    + a^6 \rho_0^{11/3}\zeta^{1/2}  (\log b/a)^2
    + a^{10} \rho_0^5\zeta^{9/2}  (\log b/a)^3
    \end{pmatrix}
    \vspace*{0.2em}
    & d= 3,
    \\
    O
    \begin{pmatrix}
    a^{4}b^{-2} \rho_0^{3}
    + a^2b^2  \rho_0^{4}\zeta^{-1}
    + a^4  \rho_0^4\zeta \log b/a
    + a^{6}  \rho_0^5\zeta^{3} (\log b/a)^3
    \end{pmatrix}
    \vspace*{0.2em}
    & d= 2,
    \\
    O
    \begin{pmatrix}
    a^{2}b^{-1} \rho_0^{4}
    + ab \rho_0^5 \log b/a 
    + a^2  \rho_0^5 \zeta^{3/2}(\log b/a)^3
    \end{pmatrix}
    & d=1.
    \end{cases}
\end{aligned}
\end{equation*}
Using that $\lim_{L\to \infty} \left[\frac{1}{L^d\beta} \log Z\right] = \psi_0(\beta,\mu)$ and 
optimising in $b$ we find for the choices (recall that we require $b\lesssim \rho_0^{-1/d}$)
\begin{equation*}
  b = \begin{cases}
  \min 
  \left\{
    a(a^3\rho_0)^{-2/15} \zeta^{1/5}, 
    \rho_0^{-1/3}
  \right\} & d=3 ,
  \\
  \min
  \left\{
    a (a^2\rho_0)^{-1/4} \zeta^{1/4}, 
    \rho_0^{-1/2}
  \right\}
  & d=2,
  \\
  a (a\rho_0)^{-1/2} \abs{\log a\rho_0}^{-1/2}
  & d=1,
  \end{cases}
\end{equation*}
that 
\begin{equation*}
\begin{aligned}
  \psi(\beta,\mu)
  & 
  \geq \psi_0(\beta,\mu) 
		- 2\pi c_d \frac{-\Li_{d/2+1}(-e^{\beta\mu})}{(-\Li_{d/2}(-e^{\beta\mu}))^{1+2/d}} a^d \rho_0^{2+2/d} 
    \left[1 + \delta_d\right],
\end{aligned}
\end{equation*}
where $\delta_d$ is as in \Cref{eqn.delta.errors.high.temp}. 
The calculations above are valid as long as the conditions of \Cref{lem.conv.GGR} are satisfied. 
That is, if $a^d  \rho_0 \zeta^{d/2} \abs{\log a^d\rho_0}$ is sufficiently small.
This concludes the proof of \Cref{lem.pressure.high.temp}.
It remains to give the proof of \Cref{lem.bdd.eps.errors}.

\subsection{Error-terms (proof of \texorpdfstring{\Cref{lem.bdd.eps.errors}}{Lemma~\ref*{lem.bdd.eps.errors}})}
\label{sec.bdd.eps.errors}
In this section we give the 
\begin{proof}[{Proof of \Cref{lem.bdd.eps.errors}}]
To better illustrate where the different error-terms come from we will write them in terms of the quantities $I_g, I_\gamma$
and $I_{|x|^ng} := \int_{\R^d} |x|^n |g(x)| \ud x$, $n\geq 1$.
By \Cref{lem.bdd.gamma.Ig.Igamma} and \Cref{eqn.Ig.bdd} we have the bounds 
\[
  I_g \leq C a^d\log b/a, \quad 
  I_\gamma\leq C \zeta^{d/2} = C (1+\abs{\log z})^{d/2}, \quad 
  I_{|x|^ng} \leq C a^d b^n.
\]
For the analysis of the error-terms we use the bounds \cite[Equation (4.13)]{Lauritsen.2023} and \cite[Equations (4.10) and (4.22)]{Lauritsen.Seiringer.2023}.
To state these, we define for any diagram $(\pi,G) \in \tilde\mcL_p^m$ the numbers 
$k = k(G) = k(\pi,G)$ as the number of clusters entirely with internal vertices (of sizes $n_1,\ldots,n_k$) and 
$\kappa = \kappa(G) = \kappa(\pi,G)$ as the number of clusters with each at least one external vertex 
(of sizes [meaning number of internal vertices] $n_1^*,\ldots,n_\kappa^*$).
Define 
\[
	n_g^* := \sum_{\lambda=1}^\kappa n_\lambda^*,
	\qquad 
	n_g := \sum_{\ell=1}^k n_\ell - 2k.
\]
The bounds \cite[Equation (4.13)]{Lauritsen.2023}, \cite[Equations (4.10) and (4.22)]{Lauritsen.Seiringer.2023} 
then read\footnote{The case $m=0$ 
is not included in the statement in \cite[Equation (4.13)]{Lauritsen.2023} and \cite[Equations (4.10) and (4.22)]{Lauritsen.Seiringer.2023}.
It follows from the analysis in \cite[Section 3.1.1]{Lauritsen.Seiringer.2023} and \cite[Section 4.1]{Lauritsen.2023}, however.}
for any $k_0, n_{g0}$
\begin{equation}\label{eqn.tail-bound}
\frac{1}{p!} \abs{\sum_{\substack{(\pi,G)\in\tilde \mcL_p^m \\ k(\pi,G) = k_0 \\ n_g(\pi,G) + n_{g}^*(\pi,G) = n_{g0}}} \Gamma_{\pi,G} }
	\leq 
	\begin{cases}
	C L^d \rho_0\left(C \rho_0I_g\right)^{n_{g0}+k_0} I_\gamma^{k_0-1}
	& m=0,
	\\
	C_m \rho_0^m \left(C\rho_0I_g\right)^{n_{g0}+k_0} I_\gamma^{k_0}
	& m > 0,
	\end{cases}
	\qquad 
	p=2k_0 + n_{g0},
\end{equation}
where the constants $C,C_m$ depend only on $m$ but not on $n_{g0}$ or $k_0$ (in particular not on $p$).

From this bound the natural ``size'' of a diagram $(\pi,G)\in \tilde\mcL_p^m$ is not $p$ but rather $n_g + n_g^* + k$,
since its value is (neglecting $\log$'s and dependence on $z$) $\lesssim\rho_0^m(a^d\rho_0)^{n_g + n_g^* + k}$.
For the bounds of the terms $\eps_2, \eps_3, \eps_Z$ we will bound sufficiently large diagrams
by the bound in \Cref{eqn.tail-bound} and do a more precise computation for small diagrams. 
We first bound $\eps_Z$.

\subsubsection{Bound of \texorpdfstring{$\eps_Z$}{epsilonZ}}
We have by \Cref{lem.conv.GGR} 
\[
	\eps_Z 
		= -\frac{1}{\beta} \log \frac{Z_J}{Z} 
		= -\frac{1}{\beta} \sum_{p=2}^\infty \frac{1}{p!} \sum_{(\pi, G)\in \mcL_p} \Gamma_{\pi,G}.
\]
We use the bound in \Cref{eqn.tail-bound} above for $m=0$ and for diagrams with $n_g + k \geq 2$.
These are precisely the diagrams with $p\geq 3$
(note that $k\geq 1$ for any diagram $(\pi,G)\in \mcL_p$).
Thus
\begin{equation*}
\begin{aligned}
\sum_{p=3}^\infty \frac{1}{p!} \abs{\sum_{\substack{(\pi,G)\in \mcL_p}} \Gamma_{\pi,G}} 
  & \leq C L^d \rho_0 \sum_{\substack{k_0 \geq 1 \\ n_{g0}+k_0\geq 2}} \left(C\rho_0 I_g\right)^{n_{g0}+k_0} I_\gamma^{k_0-1}
  \\
  & = C L^d \rho_0 \left[\sum_{n_{g0}=1}^\infty \left(C\rho_0 I_g\right)^{n_{g0}+1} 
    + \sum_{k_0=2}^\infty \sum_{n_{g0}=0}^\infty \left(C\rho_0 I_g\right)^{n_{g0}+k_0} I_\gamma^{k_0-1}\right]
  \\  & 
  \leq 
  C L^d \rho_0^3 I_g^2 (1 + I_\gamma)
\end{aligned}
\end{equation*}
for sufficiently small $\rho_0 I_g$ and $\rho_0 I_g I_\gamma$.
For the diagrams with $n_g + k=1$ we do a more precise calculation. 
These are precisely the diagrams with $p=2$. In particular these diagrams have $n_g = 0$ and $k=1$.
We have then
(recall that pictures of diagrams refer to their values)
\begin{equation*}
\begin{aligned}
\sum_{\substack{(\pi,G)\in\mcL_2}}\Gamma_{\pi,G} 
	& = 
\vcenter{\hbox{
\begin{tikzpicture}[line cap=round,line join=round,>=triangle 45,x=1.0cm,y=1.0cm]
	\node (1) at (0,1) {};
	\node (2) at (0,0) {};
	\draw[dashed] (1) -- (2);
	\draw[->] (1) to[out=-30,in=30,loop] (1);
	\draw[->] (2) to[out=-30,in=30,loop] (2);
	\foreach \i in {1,2} \draw[fill] (\i) circle [radius=1.5pt];
\end{tikzpicture}
}}
+
\vcenter{\hbox{
\begin{tikzpicture}[line cap=round,line join=round,>=triangle 45,x=1.0cm,y=1.0cm]
	\node (1) at (0,1) {};
	\node (2) at (0,0) {};
	\draw[dashed] (1) -- (2);
	\draw[->] (1) to[bend right] (2);
	\draw[->] (2) to[bend right] (1);
	\foreach \i in {1,2} \draw[fill] (\i) circle [radius=1.5pt];
\end{tikzpicture}
}}
	\\ 
	& = \iint \det \begin{bmatrix}
	\gamma^{(1)}(0) & \gamma^{(1)}(x-y) \\ \gamma^{(1)}(y-x) & \gamma^{(1)}(0)
	\end{bmatrix}
	g(x-y) \ud x \ud y
	\\
	& = \iint \rho^{(2)}(x,y) g(x-y) \ud x \ud y
	\\ 
	& = O\left(L^d I_{|x|^2g} \rho_0^{2+2/d}\right)
\end{aligned}
\end{equation*}
using \Cref{eqn.rho2.free}.
Thus, using \Cref{lem.bdd.gamma.Ig.Igamma} and recalling that $\beta \sim \zeta \rho_0^{-2/d}$ from \Cref{rmk.asym.beta},
we conclude that 
\[	
\begin{aligned}
	\frac{1}{L^d}\abs{\eps_Z}
	= \frac{1}{\beta L^d} \abs{\log \frac{Z_J}{Z}}
	& \lesssim 
  I_{|x|^2g}  \rho_0^{2+4/d}\zeta ^{-1}
  + I_{g}^2 (I_\gamma + 1) \rho_0^{3+2/d}\zeta ^{-1} 
  \\
  & 
  \lesssim
   a^db^2  \rho_0^{2+4/d}\zeta ^{-1}
  + a^{2d}  \rho_0^{3 + 2/d} \zeta ^{d/2-1}(\log b/a)^2.
\end{aligned}
\]

\subsubsection{Bound of \texorpdfstring{$\eps_3$}{epsilon3}}
We have by \Cref{lem.conv.GGR}
\[
	\rho_J^{(3)}
	= f_{12}^2 f_{13}^2 f_{23}^2
	\left[
		\rho^{(3)} + \sum_{p=1}^\infty \frac{1}{p!} \sum_{(\pi,G)\in \tilde\mcL_p^3} \Gamma_{\pi,G}^3
	\right].
\]
Note that $\rho^{(3)}$ vanishes whenever two particles are incident and it is symmetric in exchange of the particles.
Thus, for fixed $x_1$ as a function of $x_2,x_3$ it vanishes quadratically around $x_2=x_1$ and $x_3=x_1$.
Thus, by Taylor expanding $\rho^{(3)}$ in $x_2$ and $x_3$ around $x_2=x_1$ and $x_3=x_1$ we get 
$\rho^{(3)} \leq C \rho_0^{3+4/d} |x_1-x_2|^2 |x_1-x_3|^2$ using \Cref{lem.bdd.sum.gamma.k} to bound the derivatives.
We use 
the bound in \Cref{eqn.tail-bound} on the remaining terms. 
(That is, a precise calculation for diagrams with $n_g + n_g^* + k =0$ and the bound in \Cref{eqn.tail-bound}
for diagrams with $n_g + n_g^* + k \geq 1$.)
Thus, by a similar computation as for $\eps_Z$
\begin{equation*}
\begin{aligned}
\abs{\sum_{p=1}^\infty \frac{1}{p!} \sum_{(\pi,G)\in \tilde\mcL_p^3} \Gamma_{\pi,G}^3}
	& \leq C \rho_0^3 \left[
    \sum_{n_{g0}=1}^\infty (C\rho_0 I_g)^{n_{g0}} 
    + \sum_{k_0=1}^\infty \sum_{n_{g0}=0}^\infty (C\rho_0 I_g)^{n_{g0} + k_0} I_\gamma^{k_0} \right] 
  \leq C \rho_0^4 I_g (1+I_\gamma)
\end{aligned}
\end{equation*}
for sufficiently small $\rho_0 I_g$ and $\rho_0 I_g I_\gamma$.
Moreover, $f\leq 1$ and the support of $\nabla f$ is contained a ball of radius $\sim b$. 
Thus by \Cref{eqn.int.f.nabla.f.xn} and \Cref{lem.bdd.gamma.Ig.Igamma}
\begin{equation*}
\begin{aligned}			
	\abs{\eps_3}
	& \leq 
	C L^d \rho_0^{3+4/d} \left(\int f |\nabla f| |x|^2 \right)^2
	+ C L^d I_g (I_\gamma+1) \rho_0^4  \left(\int f |\nabla f|  \right)^2
	\\
	& \leq 
	C L^d a^{2d} b^2 \rho_0^{3+4/d}
	+ C L^d a^{3d-2} \rho_0^4 \zeta^{d/2} \log b/a.
\end{aligned}
\end{equation*}

\paragraph{Refined analysis in dimension $d=1$.}
In dimension $d=1$ we need also to analyse diagrams with $k +n_g + n_g^* = 1$ in more detail.
Intuitively this follow by ``counting powers of $\rho_0$'':
The claimed leading term in \Cref{thm.main} is of order $a\rho_0^4$. 
Thus, we need to compute precisely all diagrams for which the naive bound \Cref{eqn.tail-bound} 
only gives a power $\leq 4$ of $\rho_0$.

The diagrams with $k +n_g + n_g^* = 1$ have either $p=1$, in which case $n_g^*=1$, or $p=2$, in which case $k=1$.
For the diagrams with $p=1$ for any graph any permutation makes each linked component have at least one external vertex and thus 
we get 
\begin{align*}
\sum_{(\pi,G)\in \tilde\mcL_1^3} \Gamma_{\pi,G}^3 
  & = \sum_{G\in \mcG_1^3} \int \rho^{(4)} \prod_{e\in G} g_e \ud x_4.
\intertext{Bound all but one $g$-factor, by symmetry say $g_{14}$, by $|g_{ij}|\leq 1$ and Taylor expand
$\rho^{(4)}$ in $x_2,x_3,x_4$ around $x_j=x_1$ to get 
$\rho^{(4)} \leq C \rho^{10} |x_1-x_2|^2 |x_1-x_3|^2 |x_1-x_4|^2$ similarly to the bound on $\rho^{(3)}$ above.
We conclude}
  |\cdot| 
    & \leq C \rho_0^{10} |x_1-x_2|^2 |x_1-x_3|^2
  \int |g(z)| |z|^2 \ud z
  \\
  & \leq C ab^2\rho_0^{10} |x_1-x_2|^2 |x_1-x_3|^2.
\end{align*}
By \Cref{eqn.int.f.nabla.f.xn} this gives the contribution $L a^3b^{4}\rho_0^{10}$ to $\eps_3$.
For $p=2$ we have the graph (recall that $*$'s label external vertices)
\begin{equation*}
G = 
\vcenter{\hbox{
\begin{tikzpicture}[line cap=round,line join=round,>=triangle 45,x=1.0cm,y=1.0cm]
  \node (1) at (-0.5,1) {};
  \node (2) at (0.5,1) {};
  \node (3) at (1.5,1) {};
  \node (4) at (0,0) {};
  \node (5) at (1,0) {};
  \draw[dashed] (4) -- (5);
  \foreach \i in {1,...,5} \draw[fill] (\i) circle [radius=1.5pt];
  \foreach \i in {1,2,3} \node[anchor=east] at (\i) {$*$};
  \foreach \i in {1,...,5} \node[anchor=south] at (\i) {$\i$};
\end{tikzpicture}
}}
\end{equation*}
The only $\pi$'s for which $(\pi,G)\notin\tilde\mcL_2^3$ are those not connecting $\{4,5\}$ to $\{1,2,3\}$.
Thus 
\begin{equation*}
\begin{aligned}
\sum_{(\pi,G)\in \tilde\mcL_2^3} \Gamma_{\pi,G}^3 
  & = \int \left[\rho^{(5)}(x_1,\ldots,x_5) - \rho^{(3)}(x_1,x_2,x_3)\rho^{(2)}(x_4,x_5)\right] g_{45} \ud x_4 \ud x_5.
\end{aligned}
\end{equation*}
This vanishes (quadratically) whenever any $x_i$ and $x_j$, $i,j=1,2,3$ are incident.
Thus, as with $\rho^{(3)}$ and $\rho^{(4)}$, we bound the derivatives and use Taylor's theorem.
Denote the derivative w.r.t. $x_j$ by $\partial_{x_j}$. We are thus interested in bounding 
$\partial_{x_2}^2\partial_{x_3}^2\Gamma_{\pi,G}^3$.
By explicit computation (with the permutation denoted $\pi^{-1}$ for convenience of notation)
we have
\begin{align*}
  & \partial_{x_2}^2\partial_{x_3}^2 \Gamma_{\pi^{-1},G}^3
    \\ &\quad  = \partial_{x_2}^2\partial_{x_3}^2 
      \left[
        (-1)^{\pi} 
        \frac{1}{L^5} \sum_{k_1,\ldots,k_5} 
        \hat \gamma^{(1)}(k_1)\cdots \hat \gamma^{(1)}(k_5)
        \iint e^{i(k_1-k_{\pi(1)})x_1} \cdots e^{i(k_5-k_{\pi(5)})x_5} g_{45} \ud x_4 \ud x_5
      \right]
    \\
    & \quad = - (-1)^{\pi} \frac{1}{L^{4}} \sum_{k_1,\ldots,k_5} (k_2 - k_{\pi(2)})^2 (k_3 - k_{\pi(3)})^2 
      \hat \gamma^{(1)}(k_1)\cdots \hat \gamma^{(1)}(k_5)
    \\ & \qquad 
      \times 
      e^{i(k_1-k_{\pi(1)})x_1} \cdots e^{i(k_3-k_{\pi(3)})x_3} \hat g(k_4  -k_{\pi(4)}) \chi_{(k_5-k_{\pi(5)} + k_4 - k_{\pi(4)}=0)},
\intertext{where $\chi$ denotes a characteristic function. 
Any permutation such that $(\pi,G) \in \tilde\mcL_3^2$ has $\pi(\{4,5\}) \ne \{4,5\}$.
In particular for the relevant permutations the characteristic function is not identically one, and thus effectively
it reduces the number of $k$-sums by $1$.
More precisely we get for the permutations with $\pi(5),\pi(4) \ne 5$ (the others are similar)}
	& \quad = - (-1)^{\pi} \frac{1}{L^{4}} \sum_{k_1,\ldots,k_4} (k_2 - k_{\pi(2)})^2 (k_3 - k_{\pi(3)})^2 
      \hat \gamma^{(1)}(k_1)\cdots \hat \gamma^{(1)}(k_4)  \hat \gamma^{(1)}(-k_4 + k_{\pi(4)} + k_{\pi(5)})
    \\ & \qquad 
      \times 
      e^{i(k_1-k_{\pi(1)})x_1} \cdots e^{i(k_3-k_{\pi(3)})x_3} \hat g(k_4  -k_{\pi(4)}).
\end{align*}
Bounding $\abs{\hat \gamma^{(1)}(-k_4 + k_{\pi(4)} + k_{\pi(5)})} \leq 1$ and $\abs{\hat g} \leq I_g \leq C a \log b/a$ the $k$-sums 
are readily bounded by \Cref{lem.bdd.sum.gamma.k}.
Thus for any valid permutation $\pi$ we have
\begin{equation*}
\abs{\partial_{x_2}^2\partial_{x_3}^2\Gamma_{\pi,G}^3}
\leq C a \rho_0^{4+4} \log b/a.
\end{equation*}
By Taylor's theorem we conclude that 
\begin{equation*}
\abs{\Gamma_{\pi,G}^3}
\leq C a \rho_0^{4+4} \log b/a |x_1-x_2|^2 |x_1-x_3|^2.
\end{equation*}
We thus get the contribution to $\eps_3$ of 
$L a^3 b^2 \rho_0^{8} \log b/a$ by \Cref{eqn.int.f.nabla.f.xn}.
Finally, using the bound in \Cref{eqn.tail-bound} for diagrams with $k+n_g+n_g^* \geq 2$ 
we get (again for suffiently small $\rho_0 I_g$ and $\rho_0 I_g I_\gamma$)
\begin{equation*}
\sum_{p=2}^\infty \frac{1}{p!}\abs{\sum_{\substack{(\pi,G)\in \tilde\mcL_p^2 \\ (k + n_g + n_g^*)(\pi,G) \geq 2}} \Gamma_{\pi,G}^2}
  \leq C a^3 \rho_0^5 \zeta (\log b/a)^2.
\end{equation*}
By \Cref{eqn.int.f.nabla.f.xn} this gives a contribution to $\eps_3$ of $L a^2 \rho_0^5 (\log b/a)^2$.
We conclude the bound 
\begin{equation*}
\begin{aligned}
\abs{\eps_3} 
  & \leq C L \left(a^2 b^4 \rho_0^{9} + a^3 b^{4} \rho_0^{10} + a^3 b^2 \rho_0^{8} \log b/a + a^2  \rho_0^5 \zeta(\log b/a)^2\right)
  \\ & \leq CL a^2  \rho_0^5\zeta  (\log b/a)^2
\end{aligned}
\end{equation*}
in dimension $d=1$.

\subsubsection{Bound of \texorpdfstring{$\eps_2$}{epsilon2}}
We use the bound in \Cref{eqn.tail-bound} for diagrams with $n_g + n_g^* + k \geq 3$ 
and a more precise analysis for the small diagrams.
Write 
\begin{equation}\label{eqn.eps2.decompose.size}
	\sum_{p=2}^\infty \frac{1}{p!} \sum_{(\pi,G) \in \tilde\mcL_p^2} \Gamma_{\pi,G}^2
	= \xi_{=1} + \xi_{=2} + \xi_{\geq 3},
\end{equation}
where $\xi_{=j}$ is the sum of the values of all diagrams with $n_g + n_g^* + k = j$ and $\xi_{\geq 3}$
is the sum of the values of all diagrams with $n_g + n_g^* + k \geq 3$.

For the large diagrams with $n_g + n_g^* + k \geq 3$ we have 
similarly as above for $\rho_0 I_g$ and $\rho_0 I_g I_\gamma$ sufficiently small 
\begin{equation}\label{eqn.bdd.xi.large}
	\abs{\xi_{\geq 3}}
	= 
	\abs{\sum_{p=2}^\infty \frac{1}{p!} \sum_{\substack{(\pi,G)\in \tilde\mcL_p^2 \\ (k+n_g+n_g^*)(\pi,G) \geq 2}} \Gamma_{\pi,G}^2}
  \leq C \rho_0^5 I_g^3 (1 + I_\gamma^3).
\end{equation}

\paragraph{Diagrams with $k+n_g+n_g^* = 1$.}
For the diagrams with $p=1$ and $p=2$ with $k=1$ we do a more precise calculation.
For $p=1$ there are three possible $g$-graphs:
(Recall that $*$'s label the external vertices)
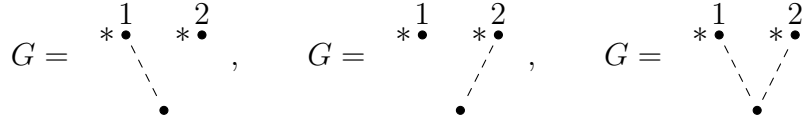
\begin{equation*}
\begin{aligned}
G=
\vcenter{\hbox{
\begin{tikzpicture}[line cap=round,line join=round,>=triangle 45,x=1.0cm,y=1.0cm]
	\node (1) at (-0.5,1) {};
	\node (2) at (0.5,1) {};
	\node (3) at (0,0) {};
	\draw[dashed] (1) -- (3);
	\foreach \i in {1,2,3} \draw[fill] (\i) circle [radius=1.5pt];
	\foreach \i in {1,2} \node[anchor=south] at (\i) {$\i$};
	\foreach \i in {1,2} \node[anchor=east] at (\i) {$*$};
\end{tikzpicture}
}}
,
\qquad 
G=
\vcenter{\hbox{
\begin{tikzpicture}[line cap=round,line join=round,>=triangle 45,x=1.0cm,y=1.0cm]
	\node (1) at (-0.5,1) {};
	\node (2) at (0.5,1) {};
	\node (3) at (0,0) {};
	\draw[dashed] (2) -- (3);
	\foreach \i in {1,2,3} \draw[fill] (\i) circle [radius=1.5pt];
	\foreach \i in {1,2} \node[anchor=south] at (\i) {$\i$};
	\foreach \i in {1,2} \node[anchor=east] at (\i) {$*$};
\end{tikzpicture}
}}
,
\qquad
G=
\vcenter{\hbox{
\begin{tikzpicture}[line cap=round,line join=round,>=triangle 45,x=1.0cm,y=1.0cm]
	\node (1) at (-0.5,1) {};
	\node (2) at (0.5,1) {};
	\node (3) at (0,0) {};
	\draw[dashed] (1) -- (3) -- (2);
	\foreach \i in {1,2,3} \draw[fill] (\i) circle [radius=1.5pt];
	\foreach \i in {1,2} \node[anchor=south] at (\i) {$\i$};
	\foreach \i in {1,2} \node[anchor=east] at (\i) {$*$};
\end{tikzpicture}
}}
\end{aligned}
\end{equation*}
Any permutation makes any of these diagrams have at least one external vertex in each linked component and thus 
\begin{align}
	\sum_{(\pi,G)\in \tilde \mcL_1^2} \Gamma_{\pi,G}^2
	& = \int \rho^{(3)}  \left[g_{13} + g_{23} + g_{13}g_{23}\right] \ud x_3
	\nonumber
\intertext{Bounding $|g_{13} g_{23}| \leq |g_{13}|$ and recalling the bound $\rho^{(3)}(x_1,x_2,x_3) \leq C \rho_0^{3+4/d} |x_1-x_2|^2 |x_1-x_3|^2$
we get by symmetry}
\abs{\sum_{(\pi,G)\in \tilde \mcL_1^2} \Gamma_{\pi,G}^2}
	& \leq C \rho_0^{3+4/d} |x_1-x_2|^2 \int |g(z)| |z|^2\ud z
	= C I_{|x|^2g} \rho_0^{3+4/d} |x_1-x_2|^2
	\label{eqn.eps2.p=1}
\end{align}

The diagrams with $p=2$ and $k=1$ have $g$-graph
\begin{equation}\label{eqn.graph.p=2.k=1}
G= 
\vcenter{\hbox{
\begin{tikzpicture}[line cap=round,line join=round,>=triangle 45,x=1.0cm,y=1.0cm]
	\node (1) at (0,1) {};
	\node (2) at (1,1) {};
	\node (3) at (0,0) {};
	\node (4) at (1,0) {};
	\draw[dashed] (3) -- (4);
	\foreach \i in {1,...,4} \draw[fill] (\i) circle [radius=1.5pt];
	\foreach \i in {1,2} \node[anchor=south] at (\i) {$\i$};
	\foreach \i in {1,2} \node[anchor=east] at (\i) {$*$};
\end{tikzpicture}
}}
\end{equation}
The only permutations $\pi$ such that $(\pi,G)\notin\tilde\mcL_2^2$ are those connecting only external to external and internal to internal,
i.e. those with either $\pi(3)=3, \pi(4)=4$ or $\pi(3)=4, \pi(4)=3$.
Thus
\begin{equation}\label{eqn.sum.diagrams.q=2.p=2.k=1}
\begin{aligned}
	\sum_{\substack{(\pi,G)\in \tilde \mcL_2^2 \\ k(\pi,G) = 1}} \Gamma_{\pi,G}^2
	& = \iint \left[\rho^{(4)}(x_1,\ldots,x_4) - \rho^{(2)}(x_1,x_2) \rho^{(2)}(x_3,x_4)\right] g_{34} \ud x_3 \ud x_4.
\end{aligned}
\end{equation}
Clearly this vanishes quadratically in $x_1-x_2$ since both determinants do, 
thus we bound it using Taylor's theorem, expanding in $x_1$ around $x_1=x_2$ 
analogously to what we did for (some of the diagrams for) $\eps_3$ above.
We treat each diagram separately. (For convenience we denote the permutation $\pi^{-1}$.)
Denoting the derivative with respect to $x_1^\mu$ by $\partial^\mu_{x_1}$ we have
\begin{equation*}
\begin{aligned}
\partial_{x_1}^\mu \partial_{x_1}^{\nu} \Gamma_{\pi^{-1},G}^2
	& = - \frac{1}{L^{4d}} \sum_{k_1,\ldots,k_4} \left(k_1^\mu - k_{\pi(1)}^\mu\right) \left(k_1^\nu - k_{\pi(1)}^\nu\right) 
		\hat \gamma^{(1)}(k_1) \hat \gamma^{(1)}(k_2) \hat \gamma^{(1)}(k_3) \hat \gamma^{(1)}(k_4) 
	\\
	& \qquad \times 
		e^{i(k_1-k_{\pi(1)})x_1} e^{i(k_2-k_{\pi(2)})x_2}
		\iint e^{i(k_3-k_{\pi(3)})x_3} e^{i(k_4-k_{\pi(4)})x_4} g(x_3-x_4) \ud x_3 \ud x_4
	\\
	& = - \frac{1}{L^{3d}} \sum_{k_1,\ldots,k_4} \left(k_1^\mu - k_{\pi(1)}^\mu\right) \left(k_1^\nu - k_{\pi(1)}^\nu\right) 
		\hat \gamma^{(1)}(k_1) \hat \gamma^{(1)}(k_2) \hat \gamma^{(1)}(k_3) \hat \gamma^{(1)}(k_4) 
	\\
	& \qquad \times 
		\hat g(k_{\pi(3)}-k_3) \chi_{(k_4-k_{\pi(4)} = k_{\pi(3)} - k_3)}
\end{aligned}		
\end{equation*}
The only permutations for which the characteristic function is identically $1$ 
are those with either $\pi(3)=3, \pi(4)=4$ or $\pi(3)=4, \pi(4)=3$.
These are exactly the permutations that do not appear in \Cref{eqn.sum.diagrams.q=2.p=2.k=1} above.
Thus, similarly as for (some of the diagrams for) $\eps_3$ above 
the charactersitic function effectively reduces the number of $k$-sums by $1$. 
Bounding $\abs{\hat g} \leq I_g$, $\hat \gamma^{(1)} \leq 1$ for one of the $\gamma^{(1)}$-factors,
and using \Cref{lem.bdd.sum.gamma.k} to bound the $k$-sums we have 
for any diagram $(\pi,G)\in \tilde\mcL_2^2$ with $G$ as in \Cref{eqn.graph.p=2.k=1} 
\begin{equation*}
\begin{aligned}
\abs{\partial_{x_1}^\mu \partial_{x_1}^{\nu} \Gamma_{\pi,G}^2}
	& \leq C I_g \rho_0^{3 + 2/d}.
\end{aligned}
\end{equation*}
We conclude the bound 
\begin{equation}\label{eqn.eps2.p=2}
\begin{aligned}
	\abs{\sum_{\substack{(\pi,G)\in \tilde \mcL_2^2 \\ k(\pi,G) = 1}} \Gamma_{\pi,G}^2}
	& \leq C I_g \rho_0^{3+2/d} |x_1-x_2|^2.
\end{aligned}
\end{equation}
In particular, by combining \Cref{eqn.eps2.p=2,eqn.eps2.p=1}, we have 
\begin{equation}\label{eqn.bdd.xi.small}
\begin{aligned}
  \abs{\xi_{=1}}
  & \leq 
  C I_g \rho_0^{3+2/d}  |x_1-x_2|^2.
\end{aligned}
\end{equation}

\paragraph{Diagrams with $k+n_g+n_g^* = 2$.}
Finally consider all diagrams with $k + n_g + n_g^* = 2$ more precisely. 
We split these into three groups.
\begin{enumerate}[(i)]
\item $n_g^* = 2$ 
\item $n_g^*=1$ and vertices $\{1\}$ and $\{2\}$ are connected
\item Remaining diagrams
\end{enumerate}
We will use a Taylor expansion to bound the values of the diagrams in group (iii).
Write 
\begin{equation*}
\xi_{=2} = \xi_{\textnormal{(i)}} + \xi_{\textnormal{(ii)}} + \xi_{\textnormal{(iii)}}
\end{equation*}
Then as $\rho_J^{(2)}(x_2;x_2)=0$ we get from \Cref{eqn.eps2.decompose.size}
\begin{equation*}
\abs{\xi_{\textnormal{(iii)}}(x_2,x_2)} 
\leq \abs{\xi_{\textnormal{(i)}}(x_2,x_2)} 
+ \abs{\xi_{\textnormal{(ii)}}(x_2,x_2)} 
+ \abs{\xi_{=1}(x_2,x_2)}
+ \abs{\xi_{\geq3}(x_2,x_2)}.
\end{equation*}
Moreover, $\xi_{\textnormal{(iii)}}$ is symmetric in exchange of $x_1$ and $x_2$ so the first order vanishes.
We conclude by Taylor's theorem that 
\begin{equation}\label{eqn.bdd.xi.taylor.expansion}
\begin{aligned}
\abs{\xi_{\textnormal{(iii)}}(x_1,x_2)} 
& \leq \abs{\xi_{\textnormal{(i)}}(x_2,x_2)} 
+ \abs{\xi_{\textnormal{(ii)}}(x_2,x_2)} 
+ \abs{\xi_{=1}(x_2,x_2)}
+ \abs{\xi_{\geq3}(x_2,x_2)} 
\\ & \quad 
+ C \sup_{\mu,\nu} \sup_{z_1,z_2} \abs{\partial_{x_1}^\mu \partial_{x_1}^\nu \xi_{\textnormal{(iii)}}(z_1,z_2)} |x_1-x_2|^2,
\end{aligned}
\end{equation}
where again $\partial_{x_1}^\mu$ denotes the derivative w.r.t. $x_1^\mu$.
Bounding $\partial_{x_1}^\mu \partial_{x_1}^\nu \xi_{\textnormal{(iii)}}$ is analogous to the argument in \cite[Proof of Lemmas 4.1 and 4.8]{Lauritsen.Seiringer.2023}:
For diagrams with an internal vertex connected to $\{1\}$ with a $g$-edge we do a precise calculation as in \cite[Proof of Lemma 4.8]{Lauritsen.Seiringer.2023}.
For the remaining diagrams where $\{1\}$ has no incident $g$-edges we modify the proof of the absolute convergence of the GGR expansion as in 
\cite[Proof of Lemma 4.1]{Lauritsen.Seiringer.2023}.

First, the diagrams in group (iii) with an internal vertex connected to $\{1\}$ with a $g$-edge all have $g$-graph 
\begin{equation}\label{eqn.graph.xi.(iii).ng*=1}
G = 
\vcenter{\hbox{
\begin{tikzpicture}[line cap=round,line join=round,>=triangle 45,x=1.0cm,y=1.0cm]
  \node (1) at (-0.5,1) {};
  \node (2) at (1.5,1) {};
  \node (3) at (0.5,1) {};
  \node (4) at (0,0) {};
  \node (5) at (1,0) {};
  \draw[dashed] (1) -- (3);
  \draw[dashed] (4) -- (5);
  \foreach \i in {1,...,5} \draw[fill] (\i) circle [radius=1.5pt];
  \foreach \i in {1} \node[anchor=east] at (\i) {$*$};
  \foreach \i in {2} \node[anchor=west] at (\i) {$*$};
  \foreach \i in {1,...,5} \node[anchor=south] at (\i) {$\i$};
\end{tikzpicture}
}}
\end{equation}
since $n_g^* = 1$ and $k+n_g+n_g^* = 2$.
Then 
\begin{equation*}
\begin{aligned}
\Gamma_{\pi^{-1},G}^2
  & = (-1)^{\pi} \frac{1}{L^{5d}} \sum_{k_1,\ldots,k_5} \hat \gamma^{(1)}(k_1) \cdots \hat \gamma^{(1)}(k_5) 
      \iiint e^{i(k_1-k_{\pi(1)}) x_1} \cdots  e^{i(k_5-k_{\pi(5)}) x_5} g_{13} g_{45} \ud x_3 \ud x_4 \ud x_5
  \\
  & = (-1)^{\pi} \frac{1}{L^{4d}} \sum_{k_1,\ldots,k_5} \hat \gamma^{(1)}(k_1) \cdots \hat \gamma^{(1)}(k_5)  
  e^{i(k_1 - k_{\pi(1)} + k_3 - k_{\pi(3)})x_1}
    e^{i(k_2-k_{\pi(2)})x_2} 
  \\ & \quad \times 
    \hat g(k_3-k_{\pi(3)}) \hat g(k_5 - k_{\pi(5)}) \chi_{(k_4 - k_{\pi(4)} + k_5 - k_{\pi(5)}=0)}.
\end{aligned}
\end{equation*}
The characteristic function $\chi$ is identically $1$ only if $\pi(\{4,5\}) = \{4,5\}$, 
but then $(\pi,G)\notin \tilde\mcL_3^2$ so these permutations do not appear in $\xi_{\textnormal{(iii)}}$.
Taking the derivative, bounding $|\hat g| \leq I_g$ and using \Cref{lem.bdd.sum.gamma.k} to bound the $k$-sums 
we conclude as above that 
\begin{equation*}\label{eqn.bdd.derivative.xi.(iii).ng*=1}
\abs{\partial_{x_1}^\mu \partial_{x_1}^\nu \Gamma_{\pi^{-1},G}^2} \leq C \rho_0^{4+2/d} I_g^2
\end{equation*}
for all diagrams $(\pi,G)\in \tilde\mcL_3^2$ with $G$ as in \Cref{eqn.graph.xi.(iii).ng*=1}.

Next, 
for the diagrams with no $g$-edges connected to $\{1\}$ the argument is as for the bound of 
$\partial_{x_1}^{\mu}\partial_{x_1}^\nu\xi_0$ in \cite[Proof of Lemma 4.1]{Lauritsen.Seiringer.2023}.
Analogously to \cite[Equations (4.19) and (4.20)]{Lauritsen.Seiringer.2023} we conclude the bound 
(the term $1$ in the factor $I_\gamma + 1$ arises similarly as in the bounds above from the value of diagrams with $k=1$)
\begin{equation*}
\abs{\partial_{x_1}^2 \sum_{\substack{(\pi,G)\in \tilde\mcL_3^2 \\ \textnormal{no $g$-edges incident to } \{1\}}} 
  \Gamma_{\pi^{-1},G}^2} 
  \leq C \rho_0^4 I_g^2 (I_\gamma + 1) 
  	\left[\rho_0^{2/d} I_\gamma + \rho_0^{1/d} I_{\partial \gamma} + I_{\partial^2 \gamma}\right] 
\end{equation*}
where 
\begin{equation*}
I_{\gamma} = \int_{[0,L]^d} \abs{\gamma^{(1)}} \ud x,
\qquad 
I_{\partial\gamma} = \max_\mu \int_{[0,L]^d} \abs{\partial^\mu \gamma^{(1)}} \ud x,
\qquad 
I_{\partial^2 \gamma} = \max_{\mu,\nu} \int_{[0,L]^d} \abs{\partial^\mu\partial^\nu \gamma^{(1)}} \ud x.
\end{equation*}
Recall that $I_\gamma \leq C \zeta^{d/2}$ by \Cref{lem.bdd.gamma.Ig.Igamma}.
By a simple modification of the proof of \Cref{lem.bdd.gamma.Ig.Igamma}
we may bound $I_{\partial\gamma} \leq C \zeta^{d/2} \rho_0^{1/d}$ and $I_{\partial^2 \gamma} \leq C \zeta^{d/2} \rho_0^{2/d}$.
Thus
\begin{equation}\label{eqn.bdd.derivative.xi.(iii)}
\abs{\partial_{x_1}^{2} \xi_{\textnormal{(iii)}}(z_1,z_2)}
\leq C \rho_0^{4+2/d} I_g^2 \zeta^{d}.
\end{equation}

\noindent
Next, we bound $\xi_{\textnormal{(i)}}$. 
For the diagrams with $n_g^*=2$, if $G$ is any graph with $n_g^*(G) = 2$ then for any permutation $\pi \in \mcS_4$ we have $(\pi,G)\in \tilde\mcL_2^2$. 
Thus Taylor expanding $\rho^{(4)}$ in $x_2,x_3,x_4$ around $x_j=x_1$ 
and bounding some $g$-factors by $1$ we get similarly to \Cref{eqn.eps2.p=1}
\begin{equation}\label{eqn.bdd.xi.(i)}
\begin{aligned}
\xi_{\textnormal{(i)}} 
  & = \sum_{\substack{G\in \mcG_2^2 \\ n_g^*(G) = 2}} \iint \rho^{(4)} \prod_{e\in G} g_e \ud x_3 \ud x_4
  \\
  \abs{\xi_{\textnormal{(i)}}}
  & \leq C\rho_0^{4+6/d} |x_1-x_2|^2 \iint |g(z_1)|^2 |g(z_2)|^2 (|z_1|^2 + |z_2|^2 + |x_1-x_2|^2)^2 
  	\ud z_1 \ud z_2
  \\
  & \leq C I_{g}^2 \rho_0^{4+6/d} |x_1-x_2|^2 (b^2 + |x_1-x_2|^2)^{2}.
\end{aligned}
\end{equation}

\noindent
Finally, we bound $\xi_{\textnormal{(ii)}}$. 
All diagrams with $n_g^* = 1$ and $\{1\}$ and $\{2\}$ connected have $g$-graph 
\begin{equation}\label{eqn.define.diagram.eps2}
G_0 = 
\vcenter{\hbox{
\begin{tikzpicture}[line cap=round,line join=round,>=triangle 45,x=1.0cm,y=1.0cm]
  \node (1) at (-0.5,1) {};
  \node (2) at (1.5,1) {};
  \node (3) at (0.5,1) {};
  \node (4) at (0,0) {};
  \node (5) at (1,0) {};
  \draw[dashed] (1) -- (3) -- (2);
  \draw[dashed] (4) -- (5);
  \foreach \i in {1,...,5} \draw[fill] (\i) circle [radius=1.5pt];
  \foreach \i in {1} \node[anchor=east] at (\i) {$*$};
  \foreach \i in {2} \node[anchor=west] at (\i) {$*$};
  \foreach \i in {1,...,5} \node[anchor=south] at (\i) {$\i$};
\end{tikzpicture}
}}
\end{equation}
For convenience of notation we denote the permutation in the diagram $\pi^{-1}$.
Then
\begin{equation*}
\begin{aligned}
& \Gamma_{\pi^{-1},G_0}^2
  \\ 
  & = (-1)^{\pi} \frac{1}{L^{5d}} \sum_{k_1,\ldots,k_5} 
    \hat\gamma^{(1)}(k_1)\cdots \hat\gamma^{(1)}(k_5) 
  \\
    & \quad \times 
    \iiint e^{i(k_1 - k_{\pi(1)})x_1} \cdots e^{i(k_5 - k_{\pi(5)})x_5} g(x_1-x_3)g(x_2-x_3) g(x_4-x_5) \ud x_3 \ud x_4 \ud x_5
  \\
  & = (-1)^{\pi} \frac{1}{L^{5d}} \sum_{k_1,\ldots,k_5} 
    \hat\gamma^{(1)}(k_1)\cdots \hat\gamma^{(1)}(k_5) 
      e^{i\left(k_1 - k_{\pi(1)} - \frac{k_3 - k_{\pi(3)}}{2}\right)x_1} 
      e^{i\left(k_2 - k_{\pi(2)} - \frac{k_3 - k_{\pi(3)}}{2}\right)x_2}
  \\
    & \quad \times 
    \int e^{i(k_3-k_{\pi(3)}) (x_3 - \frac{x_1+x_2}{2})} 
      g\left(\frac{x_1-x_2}{2} + \frac{x_1+x_2}{2} - x_3\right)
      g\left(-\frac{x_1-x_2}{2} + \frac{x_1+x_2}{2} - x_3\right)
      \ud x_3
  \\
    & \quad \times 
    \iint g(x_4-x_5) e^{i(k_4-k_{\pi(4)})(x_4-x_5)} e^{i(k_5 - k_{\pi(5)} + k_{4}-k_{\pi(4)}) x_5} \ud x_4 \ud x_5
  \\
  & = 
    (-1)^{\pi} \frac{1}{L^{4d}} \sum_{k_1,\ldots,k_5} 
    \hat\gamma^{(1)}(k_1)\cdots \hat\gamma^{(1)}(k_5) 
  \\
    & \quad \times 
      e^{i\left(k_1 - k_{\pi(1)} - \frac{k_3 - k_{\pi(3)}}{2}\right)x_1} 
      e^{i\left(k_2 - k_{\pi(2)} - \frac{k_3 - k_{\pi(3)}}{2}\right)x_2}
    \hat G_1(k_3 - k_{\pi(3)}) \hat g(k_{\pi(4)}-k_4) \chi_{(k_5 - k_{\pi(5)} + k_{4}-k_{\pi(4)} = 0)},
\end{aligned}
\end{equation*}
where
\begin{equation*}
\hat G_1(k) := \int e^{-ikz} g\left(\frac{x_1-x_2}{2} + z\right) g\left(-\frac{x_1-x_2}{2} + z\right) \ud z.
\end{equation*}

\noindent
We group together pairs of diagrams $\pi$ and (using cycle notation) $\pi \cdot (4\,5) = (\pi(4)\,\pi(5)) \cdot \pi$, meaning where $\pi(4)$ and $\pi(5)$ are swapped.
These have opposite signs.
Thus,
\begin{equation*}
\begin{aligned}
& \Gamma_{\pi^{-1},G_0}^2 + \Gamma_{(\pi (4\,5))^{-1}, G_0}
  \\ 
  & = 
    (-1)^{\pi} \frac{1}{L^{4d}} \sum_{k_1,\ldots,k_5} 
    \hat\gamma^{(1)}(k_1)\cdots \hat\gamma^{(1)}(k_5) 
      e^{i\left(k_1 - k_{\pi(1)} - \frac{k_3 - k_{\pi(3)}}{2}\right)x_1} 
      e^{i\left(k_2 - k_{\pi(2)} - \frac{k_3 - k_{\pi(3)}}{2}\right)x_2}
  \\
    & \quad \times 
    \hat G_1(k_3 - k_{\pi(3)}) 
    \chi_{(k_5 - k_{\pi(5)} + k_{4}-k_{\pi(4)} = 0)}
    \left[
    \hat g(k_{\pi(4)}-k_4) 
    - \hat g(k_{\pi(5)} - k_4)
    \right].
\end{aligned}
\end{equation*}
We Taylor expand $\hat g(k_{\pi(5)} - k_4)$ in $k_{\pi(5)}$ around $k_{\pi(5)} = k_{\pi(4)}$. That is,
\begin{equation*}
\hat g(k_{\pi(5)} - k_4) = \hat g(k_{\pi(4)} - k_4) + O\left(\nabla \hat g\right) \abs{k_{\pi(4)} - k_{\pi(5)}},
\end{equation*}
where $O\left(\nabla \hat g\right)$ should be interpreted as being bounded by 
$\abs{\nabla \hat g(k)} \leq \int |x| |g(x)| = I_{|x|g}$ uniformly in $k_{\pi(4)} - k_{\pi(5)}$.
Moreover, $\abs{\hat G_1}\leq I_g$.
Thus
\begin{equation*}
\begin{aligned}
& \abs{\Gamma_{\pi^{-1},G_0}^2 + \Gamma_{(\pi (4\,5))^{-1}, G_0}}
  \\ 
  & \leq
    C I_g I_{|x|g}
    \times
    \frac{1}{L^{4d}} \sum_{k_1,\ldots,k_5} 
    \hat\gamma^{(1)}(k_1)\cdots \hat\gamma^{(1)}(k_5) 
    \abs{k_{\pi(4)} - k_{\pi(5)}}
    \chi_{(k_5 - k_{\pi(5)} + k_{4}-k_{\pi(4)} = 0)}.
\end{aligned}
\end{equation*}
The characteristic function is not identically $1$ for linked diagrams. 
Indeed, if $\pi(\{4,5\}) = \{4,5\}$ then the diagram would not be linked.
Thus, the characteristic function effectively reduces the number of $k$-sums by $1$.
Bounding similarly as above $\hat \gamma^{(1)} \leq 1$ and 
using finally \Cref{lem.bdd.sum.gamma.k} to bound the $k$-sums we conclude 
for any permutation $\pi$ such that $(\pi,G_0)\in \tilde\mcL_3^2$ that
\begin{equation*}
\begin{aligned}
\abs{\Gamma_{\pi^{-1},G_0}^2 + \Gamma_{(\pi (4\,5))^{-1}, G_0}}
  & \leq C \rho_0^{4+1/d} I_{|x|g} I_g.
\end{aligned}
\end{equation*}
Since, $\pi$ and $\pi (4\, 5)$ either both give rise to linked diagrams or neither do we conclude that 
\begin{equation}\label{eqn.bdd.xi.(ii)}
\abs{\xi_{\textnormal{(ii)}}} 
= \frac{1}{3!}\abs{\sum_{(\pi,G_0) \in \tilde\mcL_3^2} \Gamma_{\pi,G_0}}
\leq C \rho_0^{4+1/d} I_{|x|g} I_g.
\end{equation}
Combining then \Cref{eqn.bdd.xi.(i),eqn.bdd.xi.small,eqn.bdd.xi.large,eqn.bdd.xi.taylor.expansion,eqn.bdd.derivative.xi.(iii),eqn.bdd.xi.(ii)}
and using \Cref{lem.bdd.gamma.Ig.Igamma}
we conclude the bound 
\begin{equation*}
\abs{\xi_{\textnormal{(iii)}}} 
\lesssim 
  a^{2d} b \rho_0^{4+1/d} \log b/a
  + a^{3d} \rho_0^5 \zeta^{3d/2} (\log b/a)^3
  + a^{2d} \rho_0^{4+2/d} \zeta^{d} (\log b/a)^2 |x_1-x_2|^2.
\end{equation*}

We conclude the bound  
\begin{equation*}
\begin{aligned}
  \abs{\sum_{p=1}^\infty \frac{1}{p!} \sum_{(\pi,G)\in \tilde\mcL_p^2} \Gamma_{\pi,G}^2}
  & \lesssim
    a^{d} \rho_0^{3+2/d} \log b/a |x_1-x_2|^2 + a^{d} b^2 \rho_0^{3+4/d} |x_1-x_2|^2 
  \\ & \quad 
  + a^d \rho_0^{4+6/d} |x_1-x_2|^2 (b^2 + |x_1-x_2|^2)^2 (\log b/a)^2 
  + a^{2d} b \rho_0^{4+1/d} \log b/a 
  \\ & \quad 
  + a^{2d} \rho_0^{4+2/d} \zeta^{d} (\log b/a)^2 |x_1-x_2|^2 
  + a^{3d} \rho_0^5 \zeta^{3d/2} (\log b/a)^3. 
\end{aligned}
\end{equation*}
Thus, using \Cref{lem.bdd.int.f} we get
\begin{equation*}
\begin{aligned}
\frac{\abs{\eps_2}}{L^d} 
  & \lesssim
  a^{2d} \rho_0^{3+2/d} \log b/a 
  + a^{2d} b^{2} \rho_0^{3+4/d} 
  + a^{4d} b^{4 - d} \rho_0^{4+6/d} (\log b/a)^2%
  + a^{3d-2} b \rho_0^{4+1/d} \log b/a 
  \\ & \quad 
  + a^{3d} \rho_0^{4+2/d} \zeta^{d} (\log b/a)^2
  + a^{4d-2} \rho_0^5 \zeta^{3d/2} (\log b/a)^3.
  \\
  & \lesssim 
  \begin{cases}
  a^{2d} \rho_0^{3+2/d} \log b/a + a^{4d-2} \rho_0^5 \zeta^{3d/2} (\log b/a)^3
  & d\geq 2, \\
  a b \rho_0^5 \log b/a + a^2  \rho_0^5 \zeta^{3/2}(\log b/a)^3
  & d = 1.
  \end{cases}
\end{aligned}
\end{equation*}
This concludes the proof of \Cref{lem.bdd.eps.errors}.
\end{proof}

\begin{remark}[{Necessity of precise analysis of diagrams with $k+n_g+n_g^*=2$}]
For the bound of $\eps_2$ we give here a precise analysis of the diagrams with $k+n_g+n_g^*=2$. 
In general, one should not expect this to be needed in dimensions $d=2,3$.
More precisely, by just considering powers of $\rho_0$, one would expect that 
diagrams with $k+n_g+n_g^* \geq 1$ are all subleading as they carry a higher power of 
$\rho_0$ (using \Cref{eqn.tail-bound}) than the claimed leading term, with exponent $2+2/d$.

The reason we need a precise analysis here is the temperature dependence of our bounds:
For some regime of temperatures the bound one would get by using \Cref{eqn.tail-bound} is not good enough.
\end{remark}

\begin{remark}[Optimality of the error bounds]\label{rmk.compare.zero.temp.d=1}
One should not expect the bound given in \Cref{eqn.bdd.xi.(ii)} to be optimal. 
More precisely in \Cref{eqn.bdd.xi.(ii)} we only took into account the cancellations of pairs of diagrams.
However, one should expect much more cancellations. 
We have 
\begin{equation*}
	\xi_{\textnormal{(ii)}}
	= \frac{1}{3!} \iiint \left[\rho^{(5)}(x_1,\ldots,x_5) - \rho^{(3)}(x_1,x_2,x_3)\rho^{(2)}(x_4,x_5)\right] 
		g_{13} g_{23} g_{45} \ud x_3 \ud x_4 \ud x_5.
\end{equation*}
Naively, just using that $\rho^{(5)}(x_1,\ldots,x_5) - \rho^{(3)}(x_1,x_2,x_3)\rho^{(2)}(x_4,x_5)$
vanishes whenever any two of the particles $1,2,3$ or the particles $4,5$ are incident 
we get by Taylor expansion 
\begin{equation}
	\abs{\rho^{(5)}(x_1,\ldots,x_5) - \rho^{(3)}(x_1,x_2,x_3)\rho^{(2)}(x_4,x_5)}
	\leq C \rho_0^{5+6/d} |x_1-x_2|^2 |x_1-x_3|^2 |x_4-x_5|^2.
\end{equation}
Using this bound and bounding $|g_{23}|\leq 1$ we get 
\begin{equation}
	\abs{\xi_{\textnormal{(ii)}}}
		\leq \rho_0^{5+6/d} a^{2d}b^4 L^d |x_1-x_2|^2.
\end{equation}
This bound is too large by a volume factor. 
(This arises since we ``forget'' that the relevant diagrams are linked when we do the Taylor expansion.)
It however illustrates how many more cancellations between the different permutations are present 
than what we used in the bound \Cref{eqn.bdd.xi.(ii)} --- it carries a higher power of $\rho_0$.
Using these cancellations but losing the information that diagrams are linked is what we did in \cite{Lauritsen.Seiringer.2023}.

If one could somehow see these cancellations, while still keeping the information
that the diagrams have to be linked, one might be able to improve upon the bound \Cref{eqn.bdd.xi.(ii)}. 
In $1$ dimension this error term is actually (for some regime of temperatures) the dominant error term.
Thus, by improving the analysis of these diagrams, one might improve the error term in \Cref{lem.pressure.high.temp}
in $d=1$.
\end{remark}

\appendix

\section{Particle density of the trial state}
\label{sec.density.Gamma_J}
In this section we give the 
\begin{proof}[Proof of \Cref{eqn.density.Gamma_J}]
We calculate $\expect{\mcN}_J$ and compare it to $\expect{\mcN}_0 = \rho_0L^d$.
We have by \Cref{eqn.rhoq_J}
\begin{equation*}
\begin{aligned}
\expect{\mcN}_J
	& = \int \rho_J^{(1)}(x) \ud x 
	= L^d \left[\rho^{(1)} + \sum_{p=1}^\infty \frac{1}{p!} \sum_{(\pi,G)\in \mcL_p^1} \Gamma_{\pi,G}^1\right]
	= \expect{\mcN}_0 + L^d \sum_{p=1}^\infty \frac{1}{p!} \sum_{(\pi,G)\in \mcL_p^1} \Gamma_{\pi,G}^1.
\end{aligned}
\end{equation*}
Next, we bound $\sum_{p=1}^\infty \frac{1}{p!} \sum_{(\pi,G)\in \mcL_p^1} \Gamma_{\pi,G}^1$.
We use the bound in \Cref{eqn.tail-bound}
for diagrams with $k+n_g+n_g^* \geq 2$, i.e.,
for $p=2$ with $k = 0, n_{g}^* = 2$ and 
for $p\geq 3$. 
That is, 
\begin{equation*}
\frac{1}{2!} \abs{\sum_{\substack{(\pi,G)\in\mcL_2^1 \\ k(\pi,G) = 0}} \Gamma_{\pi,G}^1 }
  \leq C I_g^2 
  \rho_0^3,
	\qquad
\sum_{p=3}^\infty
\frac{1}{p!} \abs{\sum_{\substack{(\pi,G)\in\mcL_p^1}} \Gamma_{\pi,G}^1 }
  \leq C I_g^2 (1 + I_\gamma^2) \rho_0^3
\end{equation*}
for sufficiently small $\rho_0 I_g$ and $\rho_0 I_g I_\gamma$.
Thus, we get 
\begin{equation*}
\sum_{p=1}^\infty \frac{1}{p!} \sum_{\substack{(\pi,G)\in\mcL_p^1}}\Gamma_{\pi,G}^1
 	=  \sum_{\substack{(\pi,G)\in\mcL_1^1}}\Gamma_{\pi,G}^1
 		+ \frac{1}{2} \sum_{\substack{(\pi,G)\in \mcL_2^1 \\ k(\pi,G) = 1}} \Gamma_{\pi,G}^1
    + O\left(I_g^2\left(I_\gamma^2+1\right) \rho_0^3\right).
\end{equation*}
For the $p=1$-term there are two diagrams. Thus (where $*$ labels the external vertex)
\begin{equation*}
\displaystyle \sum_{\substack{(\pi,G)\in\mcL_1^1}}\Gamma_{\pi,G}^1 = 
\vcenter{\hbox{
\begin{tikzpicture}[line cap=round,line join=round,>=triangle 45,x=1.0cm,y=1.0cm]
	\node (1) at (0,1) {};
	\node (2) at (0,0) {};
	\draw[dashed] (1) -- (2);
	\draw[->] (1) to[out=-30,in=30,loop] (1);
	\draw[->] (2) to[out=-30,in=30,loop] (2);
	\foreach \i in {1,2} \draw[fill] (\i) circle [radius=1.5pt];
	\foreach \i in {1} \node[anchor=south] at (\i) {$*$};
\end{tikzpicture}
}}
+
\vcenter{\hbox{
\begin{tikzpicture}[line cap=round,line join=round,>=triangle 45,x=1.0cm,y=1.0cm]
	\node (1) at (0,1) {};
	\node (2) at (0,0) {};
	\draw[dashed] (1) -- (2);
	\draw[->] (1) to[bend right] (2);
	\draw[->] (2) to[bend right] (1);
	\foreach \i in {1,2} \draw[fill] (\i) circle [radius=1.5pt];
	\foreach \i in {1} \node[anchor=south] at (\i) {$*$};
\end{tikzpicture}
}}
	= \int \det \begin{bmatrix}
	\gamma^{(1)}(0) & \gamma^{(1)}(x) \\ \gamma^{(1)}(x) & \gamma^{(1)}(0)
	\end{bmatrix}
	g(x) \ud x
	= O\left(I_{|x|^2g}\rho_0^{2+2/d}\right).
\end{equation*}
For the $p=2$-term with $k=1$ there are $4$ diagrams. Thus 
\begin{equation*}
\begin{aligned}
	\frac{1}{2}\sum_{\substack{(\pi,G)\in \mcL_2^1 \\ k(\pi,G) = 1}} \Gamma_{\pi,G}^1 
	& = 
\frac{1}{2}
\left[
\vcenter{\hbox{
\begin{tikzpicture}[line cap=round,line join=round,>=triangle 45,x=1.0cm,y=1.0cm]
	\node (1) at (0,1) {};
	\node (2) at (-0.5,0) {};
	\node (3) at (0.5,0) {};
	\draw[dashed] (2) -- (3);
	\draw[->] (1) to[bend right] (2);
	\draw[->] (2) to[bend right] (1);
	\draw[->] (3) to[out=60,in=120,loop] (3);
	\foreach \i in {1,2,3} \draw[fill] (\i) circle [radius=1.5pt];
	\foreach \i in {1} \node[anchor=south] at (\i) {$*$};
\end{tikzpicture}
}}
	+
\vcenter{\hbox{
\begin{tikzpicture}[line cap=round,line join=round,>=triangle 45,x=1.0cm,y=1.0cm]
	\node (1) at (0,1) {};
	\node (2) at (-0.5,0) {};
	\node (3) at (0.5,0) {};
	\draw[dashed] (2) -- (3);
	\draw[->] (1) to[bend right] (3);
	\draw[->] (3) to[bend right] (1);
	\draw[->] (2) to[out=60,in=120,loop] (2);
	\foreach \i in {1,2,3} \draw[fill] (\i) circle [radius=1.5pt];
	\foreach \i in {1} \node[anchor=south] at (\i) {$*$};
\end{tikzpicture}
}}
	+
\vcenter{\hbox{
\begin{tikzpicture}[line cap=round,line join=round,>=triangle 45,x=1.0cm,y=1.0cm]
	\node (1) at (0,1) {};
	\node (2) at (-0.5,0) {};
	\node (3) at (0.5,0) {};
	\draw[dashed] (2) -- (3);
	\draw[->] (1) to (2);
	\draw[->] (2) to[bend right] (3);
	\draw[->] (3) to (1);
	\foreach \i in {1,2,3} \draw[fill] (\i) circle [radius=1.5pt];
	\foreach \i in {1} \node[anchor=south] at (\i) {$*$};
\end{tikzpicture}
}}
	+
\vcenter{\hbox{
\begin{tikzpicture}[line cap=round,line join=round,>=triangle 45,x=1.0cm,y=1.0cm]
	\node (1) at (0,1) {};
	\node (2) at (-0.5,0) {};
	\node (3) at (0.5,0) {};
	\draw[dashed] (2) -- (3);
	\draw[->] (1) to (3);
	\draw[->] (3) to[bend left] (2);
	\draw[->] (2) to (1);
	\foreach \i in {1,2,3} \draw[fill] (\i) circle [radius=1.5pt];
	\foreach \i in {1} \node[anchor=south] at (\i) {$*$};
\end{tikzpicture}
}}
\right]
	\\
	& = \frac{1}{L^{3d}} \sum_{k_1,k_2,k_3} \iint \ud x_2 \ud x_3 \, \hat\gamma^{(1)}(k_1) \hat\gamma^{(1)}(k_2) \hat\gamma^{(1)}(k_3) g(x_2-x_3)
	\\ & \qquad \times 
		\left[e^{i(k_1-k_2)(x_1-x_2)} - e^{ik_1(x_1-x_2)} e^{ik_2(x_2-x_3)} e^{ik_3(x_3-x_1)}\right]
	\\
	& = \frac{1}{L^{2d}} \sum_{k_1,k_2,k_3} \hat\gamma^{(1)}(k_1) \hat\gamma^{(1)}(k_2) \hat\gamma^{(1)}(k_3) \hat g(k_1-k_2) \left[\chi_{(k_1=k_2)} - \chi_{(k_1=k_3)}\right]
	\\
	& = \frac{1}{L^{2d}} \sum_{k,\ell} \hat \gamma^{(1)}(k)^2 \hat\gamma^{(1)}(\ell) \left[\hat g(0) - \hat g(k-\ell)\right].
\end{aligned}
\end{equation*}
Taylor expanding $\hat g$ and using that $\int x g(x) = 0$ so $\nabla \hat g(0) = 0$ 
we get 
\begin{equation*}
\abs{\frac{1}{2}\sum_{\substack{(\pi,G)\in \mcL_2^1 \\ k(\pi,G) = 1}} \Gamma_{\pi,G}^1 }
  \leq C I_{|x|^2g} \rho_0^{2+2/d}.
\end{equation*}
Thus, by \Cref{lem.bdd.gamma.Ig.Igamma}
\begin{equation*}
\begin{aligned}
\abs{\sum_{p=1}^\infty \frac{1}{p!} \sum_{\substack{(\pi,G)\in\mcL_p^1}}\Gamma_{\pi,G}^1}
  &
  \leq C I_{|x|^2g} \rho_0^{2+2/d} + C I_g^2 I_\gamma^2 \rho_0^3 + C I_g^2 \rho_0^3
  \\ &
 	\leq C a^d b^2 \rho_0^{2+2/d} + C a^{2d} \rho_0^3 \zeta^d (\log b/a)^2.
\end{aligned}
\end{equation*}
That is, 
\Cref{eqn.density.Gamma_J} is satisfied.
\end{proof}

\printbibliography

\end{document}